\documentclass[10pt, journal,compsoc]{IEEEtran}
\usepackage{verbatim}

\usepackage{multirow}
\usepackage{float}
\usepackage{subfigure}
\usepackage{enumerate}
\usepackage{enumitem}
\usepackage{stfloats}
\usepackage{hyperref}
\usepackage{breakurl}
\usepackage{diagbox}
\usepackage{bbding}
\usepackage{titlesec}
\usepackage{bm}
\usepackage{verbatimbox}
\usepackage{colortbl}
\usepackage{booktabs}
\usepackage{balance}
\usepackage{makecell}
\usepackage{titlesec}
\usepackage{amsmath,amsfonts,amsthm}
\usepackage{color}
\usepackage{graphicx}
\usepackage{caption}
\usepackage{xcolor}
\usepackage{colortbl}
\usepackage{amsthm}
\usepackage{lipsum}

\usepackage[linesnumbered,ruled,vlined]{algorithm2e}
\usepackage{cleveref}

\def\BibTeX{{\rm B\kern-.05em{\sc i\kern-.025em b}\kern-.08em
    T\kern-.1667em\lower.7ex\hbox{E}\kern-.125emX}}
\usepackage{balance}

\hypersetup{
    colorlinks=true,
    linkcolor=black,
    urlcolor=black
}

\newtheorem{thm}{Theorem}

\crefname{section}{§}{§§}
\Crefname{section}{§}{§§}

\newcolumntype{L}{@{}>{\kern\tabcolsep}l<{\kern\tabcolsep}}

\newcommand{\tabincell}[2]{\begin{tabular}{@{}#1@{}}#2\end{tabular}}
\def\model{BGCH+}

\newenvironment{sequation}{\begin{equation}\setlength\abovedisplayskip{3pt}\setlength\belowdisplayskip{3pt}}{\end{equation}}

\makeatletter
\newcommand\notsotiny{\@setfontsize\notsotiny\@vipt\@viipt}
\makeatother

\def\cyk{\color{black} }
\def\cykk{\color{black} }
\def\newfig{\color{black} }
\def\drop{\color{black} }

\def\emb{\boldsymbol} 

\DeclareMathOperator\sign{sign}
\DeclareMathOperator\argmin{argmin}

\SetKwComment{Comment}{$\triangleright$\ }{}

\begin{document}



\title{Towards Effective Top-N Hamming Search via Bipartite Graph Contrastive Hashing}


\author{Yankai~Chen,
        Yixiang~Fang,
        Yifei Zhang,
        Chenhao~Ma,
        Yang~Hong,
        and~Irwin~King,~\IEEEmembership{Fellow,~IEEE}%
\IEEEcompsocitemizethanks{\IEEEcompsocthanksitem Yankai Chen, Yifei Zhang, Irwin King are with Department
of Computer Science and Engineering, The Chinese University of Hong Kong.
E-mail: \{ykchen,yfzhang,king\}@cse.cuhk.edu.hk.
\IEEEcompsocthanksitem Yang Hong is with Faculty of Information Science and Engineering, Ocean University of China.  E-mail: hongyang@stu.ouc.edu.cn.
\IEEEcompsocthanksitem Yixiang Fang, Chenhao Ma are with School of Data Science, The Chinese University of Hong Kong, Shenzhen. 
E-mail: \{fangyixiang,machenhao\}@cuhk.edu.cn}.
}

\IEEEtitleabstractindextext{%
{\cyk
\begin{abstract}
Searching on bipartite graphs serves as a fundamental task for various real-world applications, such as recommendation systems, database retrieval, and document querying.
Conventional approaches rely on similarity matching in continuous Euclidean space of vectorized node embeddings. 
To handle intensive similarity computation efficiently, hashing techniques for graph-structured data have emerged as a prominent research direction.
However, despite the retrieval efficiency in Hamming space, previous studies have encountered \textit{catastrophic performance decay}. 
To address this challenge, we investigate the problem of hashing with Graph Convolutional Network for effective Top-N search. 
Our findings indicate the learning effectiveness of incorporating hashing techniques within the exploration of bipartite graph reception fields, as opposed to simply
treating hashing as post-processing to output embeddings.
To further enhance the model performance, we advance upon these findings and propose \textbf{B}ipartite \textbf{G}raph \textbf{C}ontrastive \textbf{H}ashing (\textbf{\model}). 
\model~introduces a novel dual augmentation approach to both \textit{intermediate information} and \textit{hash code outputs} in the latent feature spaces, thereby producing more expressive and robust hash codes within a dual self-supervised learning paradigm.
Comprehensive empirical analyses on six real-world benchmarks validate the effectiveness of our dual feature contrastive learning in boosting the performance of \model~compared to existing approaches.
\end{abstract}

}

}

\maketitle

\IEEEdisplaynontitleabstractindextext
\IEEEpeerreviewmaketitle

\section{{Introduction}}

Bipartite graphs are ubiquitous in the real world for the ease of modeling various Web applications, e.g., as shown in Figure~\ref{fig:intro}(a), user-product recommendation~\cite{ma2020probabilistic,zhang2019star}, and online query-document matching~\cite{zhang2019doc2hash}.
The fundamental task of \textit{Top-N search} involves selecting the best-matched graph nodes for a given query node, enabling effective information filtering.
Machine learning advancements have popularized the use of vectorized representations (\textit{a.k.a.} embeddings) for similarity matching~\cite{grover2016node2vec,cheng2018aspect}, with \textit{Graph Convolutional Networks} (GCNs) standing out for their remarkable performance in capturing high-order connection information and enriching node embeddings~\cite{graphsage,lightgcn}. 
In addition to embedding informativeness, addressing computation latency and memory overhead is crucial for practical application deployment. 
Recently, \textit{learning to hash}~\cite{wang2017survey,jegou2010product} recently provides an alternative option to graph-based models for optimizing the model scalability.
Generally, it learns to convert continuous values into the finite binarized hash codes. 
In lieu of using \textit{full-precision}\footnote{\scriptsize The term ``full-precision'' generally refers to single-precision and double-precision. And we use float32 by default throughout this work for illustration.} 
embeddings, this approach offers space reduction and computation acceleration for Top-N object matching and retrieval in the Hamming space, providing scalability amidst the context of explosive data growth.

\begin{figure}[tp]
\begin{minipage}{0.5\textwidth}
\includegraphics[width=3.55in]{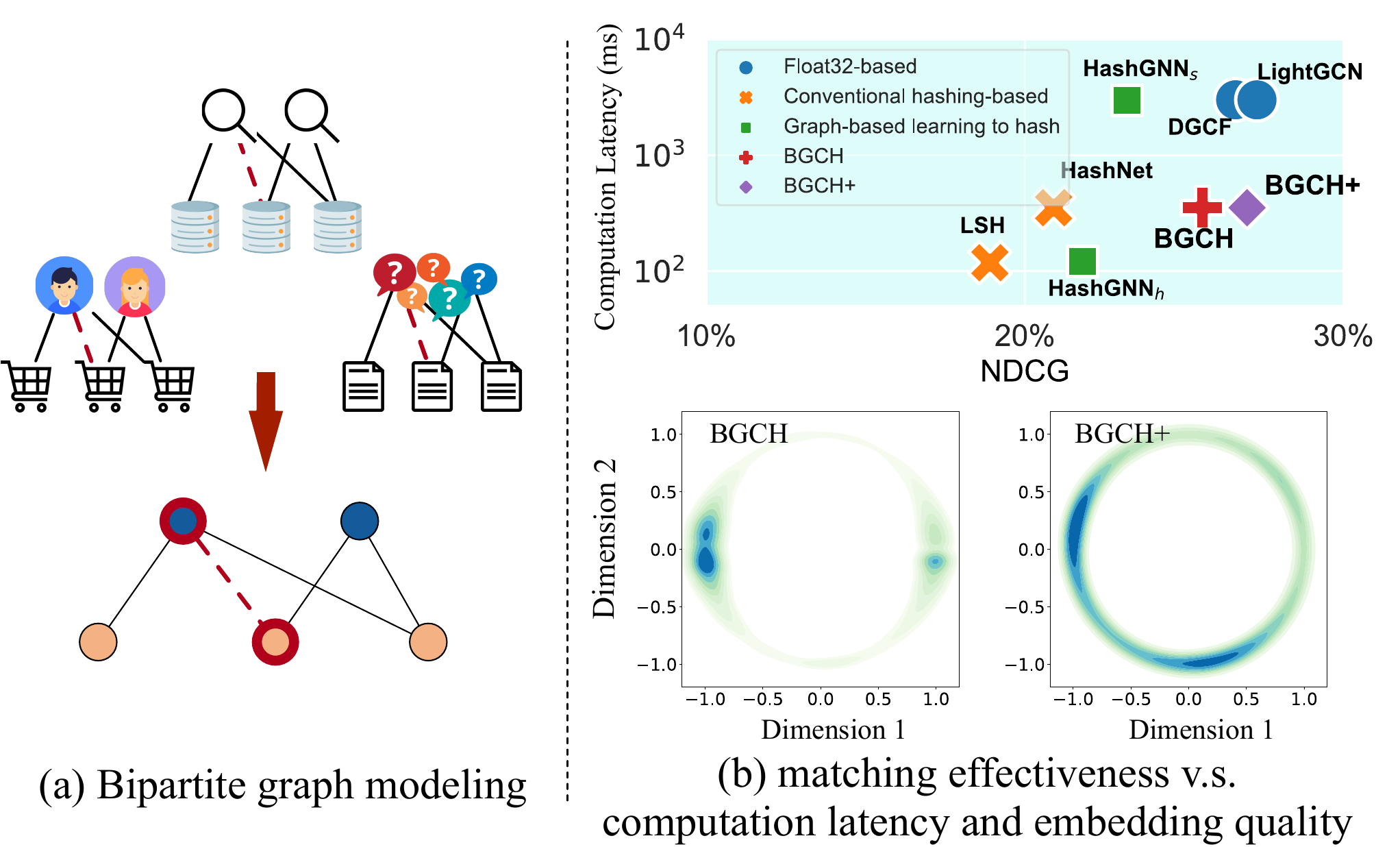}
\end{minipage} 
\caption{Bipartite graph modeling and overall model performance in terms of evaluation metrics and embedding distribution visualization (best view in color).}
\label{fig:intro}
\end{figure}

{\cyk
Despite the promising advantages of bridging GCNs and learning to hash,
simply stacking these two techniques is trivial and thus falls short of performance satisfaction.
Firstly, compared to continuous embeddings, hash codes with the same vector dimension are naturally \textit{less expressive} with finite encoding permutation in Hamming space (e.g., $2^d$ if the dimension is $d$).
\textit{Consequently, this not only leads to a coarse-grained information encoding of the graph nodes, but further derives inaccurate estimation of the pairwise node similarity}.
{\cykk Therefore, the model exhibits a conspicuous performance decay in Top-N matching and ranking.}
Secondly, for $O(1)$ complexity encoding, $\sign(\cdot)$ function is usually adopted in recent work~\cite{qin2020forward,rastegari2016xnor,lin2017towards,hashgnn}.
Despite the simplicity, optimizing neural networks with $\sign(\cdot)$ is not easy, as $\sign(\cdot)$ is not differentiable at 0 and its derivatives are 0 anywhere else.
Previous models usually use \textit{visually similar} but not necessarily \textit{theoretically relevant} functions, e.g., $\tanh(\cdot)$, for gradient estimation.
This may lead to inconsistent optimization directions between forward and backward propagation.
Their associated loss landscapes are usually steep and bumping~\cite{bai2020binarybert}, which further increases the difficulty in optimization.
These factors jointly lead to an intractable model training process.

To tackle these aforementioned challenges, we have progressively studied the problem of learning to hash with GCN on bipartite graphs in~\cite{chen2023bipartite}.
We identify the critical effect of mining the high-order correlation knowledge in bipartite graph exploration and hashing.
In response to this, {\cykk we have developed an effective learning framework namely BGCH (short for \textbf{B}ipartite \textbf{G}raph \textbf{C}onvolutional \textbf{H}ashing)~\cite{chen2023bipartite}}.
Generally, the primary objective of BGCH is to enhance the expressivity of the learned hash codes while ensuring an accordant and robust model optimization flow.
By leveraging BGCH, the empirical results demonstrate notable performance superiority and competitiveness when compared to both hashing-based and full-precision-based counterparts, respectively.
Therefore, BGCH strikes a delicate balance between matching accuracy and computational efficiency, making it a desirable solution for real-world scenarios with limited computation resources.

{\cykk
However, there may still exist two major inadequacy of BGCH~\cite{chen2023bipartite}.
Firstly, due to the discreteness of hash codes, the learning of BGCH~\cite{chen2023bipartite} in generating informative hash codes is characterized by unpleasant inefficiency, as it typically necessitates substantial training iterations for convergence.
Secondly, BGCH currently focuses on alleviating the information loss during the process of graph convolutional hashing, but ignores another fundamental point in learning quality of graph information extraction, especially for sparse graphs. As we include six real-world datasets, their graph densities vary from 0.0419, 0.00084, 0.00267, to 0.0013, 0.00062, and 0.02210.
Since the graph convolutions learn and extract the topological information, these sparse graphs may however provide limited structural knowledge as the supervision signal of model learning, resulting in a less effective subsequent graph hashing.  
}
To overcome these inadequacies, we seek to absorb the self-supervised learning capability by introducing additional regularization signals to BGCH.
Notably, we expect to empower BGCH with the contrastive learning methodology~\cite{he2020momentum,chen2020simple,wu2021self} that can extract meaningful and discriminative features from unlabeled data and regularize representations with good generalization capabilities.
The conventional way to apply contrastive learning on graphs is to first obtain different views of graph structures by explicit data augmentation, e.g., stochastic dropout of nodes, edges, or subgraphs with a certain probability~\cite{wu2021self, you2020graph, zhu2021graph}.
Then the learning objective is to maximize the consistency of the same samples under different views and distinguish different samples simultaneously.
While it seems straightforward to apply similar augmentation techniques to our topic of interest, such stochastic data manipulation may however introduce undesired variations to the learning process, as graph hashing is more sensitive to such structure variations due to its inherent discreteness property.

To provide a more robust contrastive learning paradigm for graph hashing, in this work we put forward the model \textbf{B}ipartite \textbf{G}raph \textbf{C}ontrastive \textbf{H}ashing (\model).
\model~builds upon the topology-aware hashing paradigm of BGCH while proposing a novel \textbf{dual feature contrastive learning} framework.
Specifically, different from the conventional augmentation to input data, we propose to conduct feature-level augmentations to both \textit{intermediate information} and \textit{binarized hash codes} when training.
Via random noise addition, \model~generates perturbation between contrastive views while maintaining the noise with controlled magnitude.
This feature-based augmentation approach, as opposed to explicit manipulation of the graph structure, is thus more flexible and efficient to implement.
Subsequently, these two streams of augmentations are employed within a dual feature contrastive learning objective, incorporating distinct learning granularities.
As shown in the lower section of Figure~\ref{fig:intro}(b), this learning mechanism of \model~eventually enhances the \textit{distribution uniformity} of the resultant target hash codes, thereby improving their robustness when confronted with structural variations.

Based on the well-learned hash codes, \model~substantially improves the performance in Top-N Hamming space retrieval.
The quality-cost trade-off is summarized in the upper section of Figure~\ref{fig:intro}(b), where \model~is compared to several representative counterparts, including float32-based and hashing-based models.
The evaluation is conducted on a real-world bipartite graph with over 10 million observed edges, and further experimental details can be found in~\cref{sec:exp_setup}.
Notably, as the figure lower-right corner indicates the ideal optimal performance, \model~surpasses BGCH and even achieves a comparable prediction accuracy to existing full-precision models, while still maintaining over 8$\times$ computation acceleration. 
{\cykk 
To summarize, our early model BGCH~\cite{chen2023bipartite} has the following contributions:
\begin{itemize}[leftmargin=*]
\item BGCH studies the problem of learning to hash on bipartite graphs with graph representation learning, and proposes an effective approach for effective and efficient Top-N search in Hamming space.

\item BGCH provides both theoretical and empirical effectiveness on several datasets and demonstrates efficiency in both time and space complexity.
\end{itemize}
Extending these early findings, our advanced model \model~further presents the new contributions as follows:
\begin{itemize}[leftmargin=*]
\item \model~focuses on improving the learning quality of graph convolutional hashing via leveraging self-supervised learning. To the best of our knowledge, \model~is the first to elucidate benefits of contrastive learning for graph hashing.

\item \model~proposes a novel dual feature contrastive learning paradigm that operates on feature augmentations of intermediate information and output hash codes, which is different with the conventional manner of complicated structural manipulation. 

\item We conduct a comprehensive evaluation on six real-world datasets. The empirical analyses demonstrate the efficacy of learning high-quality hash codes and its superiority in surpassing its predecessor BGCH and other counterparts.
\end{itemize}
}

We organize this paper as follows. 
We introduce the preliminaries in~\cref{sec:pre} and formally present \model~methodology in~\cref{sec:method}.
The complexity analysis is conducted in~\cref{sec:complexity}.
Then we detail all experiments and review the related work in~\cref{sec:exp} and~\cref{sec:work} with the conclusion in~\cref{sec:con}.
}

\section{{Preliminaries and Problem Formulation}}
\label{sec:pre}

{\textbf{Graph Convolution Network (GCN).}}
The general idea of GCN is to learn node embeddings by \textit{iteratively propagating and aggregating} latent features of node neighbors via the graph topology~\cite{wu2019simplifying,lightgcn,kipf2016semi}:
\begin{sequation}
\boldsymbol{V}_x^{(l)} = AGG\left(\boldsymbol{V}_x^{(l-1)}, \{\boldsymbol{V}_z^{(l-1)}: z \in \mathcal{N}(x)\}\right),
\end{sequation}%
where {\small$\emb{V}_x^{(l)} \in \mathbb{R}^d$} denotes node $x$'s embedding after $l$-th iteration of graph convolutions, indexed in the embedding matrix {\small $\emb{V}$}. 
{\small $\mathcal{N}(x)$} is the set of $x$'s neighbors.
Function $AGG(\cdot, \cdot)$ is the information aggregation function, with several implementations in previous work~\cite{kipf2016semi,graphsage,gat,xu2018powerful}, mainly aiming to transform the center node feature and the neighbor features.
In this work, we adopt the state-of-the-art graph convolution paradigm~\cite{lightgcn}.

\textbf{Bipartite Graph and Adjacency Matrix.} 
The bipartite graph is denoted as $\mathcal{G}=\{\mathcal{V}_1, \mathcal{V}_2, \mathcal{E}\}$, where $\mathcal{V}_1$ and  $\mathcal{V}_2$ are two \textit{disjoint} node sets and $\mathcal{E}$ is the set of edges between nodes in $\mathcal{V}_1$ and $\mathcal{V}_2$.
We can use $\emb{Y} \in \mathbb{R}^{|\mathcal{V}_1|\times |\mathcal{V}_2}|$ to indicate the edge transactions, where 1-valued entries, i.e., $\emb{Y}_{x,y}=1$, indicate there is an observed edge between nodes $x\in \mathcal{V}_1$ and $y \in \mathcal{V}_2$, otherwise $\emb{Y}_{x,y}=0$.
Then the adjacency matrix $\emb{A}$ of the whole graph can be defined as:
\begin{sequation}
\emb{A} = 
\begin{bmatrix}
\emb{0} & \emb{Y} \\
\emb{Y}^T & \emb{0}
\end{bmatrix}.
\end{sequation}%

\textbf{Problem Formulation.}
Given a bipartite graph $\mathcal{G}$ $=$ $\{\mathcal{V}_1, \mathcal{V}_2, \mathcal{E}\}$ along with its adjacency matrix $\emb{A}$, we devote to learn a hashing function:
\begin{sequation}
F(\emb{A} | \Theta) \rightarrow \mathcal{Q},
\end{sequation}%
where $\Theta$ is the set of all learnable parameters and embeddings and $F$ maps nodes into the $d$-dimensional Hamming space. 
Given two nodes in the bipartite graph, e.g., $x\in \mathcal{V}_1$ and $y \in \mathcal{V}_2$, their hash codes are $\mathcal{Q}_x$ and $\mathcal{Q}_y$.
Then the probability of edge existence $\widehat{\emb{Y}}_{x,y}$ between nodes $x\in \mathcal{V}_1$ and $y \in \mathcal{V}_2$ can be effectively and efficiently measured by the hash codes $\mathcal{Q}_x$ and  $\mathcal{Q}_y$, i.e., $\widehat{\emb{Y}}_{x,y}$ = $f(\mathcal{Q}_x, \mathcal{Q}_y)$ where $f$ is a score function.
Intuitively, the larger value $\widehat{\emb{Y}}_{x,y}$ is, the more likely $x$ and $y$ are matched, i.e., an edge between $x$ and $y$ exists.
We use bold uppercase and calligraphy characters for matrices and sets. The non-bolded denote graph nodes or scalars. Explanations of key notations used in this paper are attached in Table~\ref{tab:notation}.

\begin{table}[t]
\newfig
\caption{\newfig Notations and meanings.}
\label{tab:notation}
  \footnotesize
  \begin{tabular}{c|l} 
     \hline
          {\bf Notation} & {\bf Meaning}\\
     \hline\hline
          {\notsotiny $\mathcal{G},\mathcal{V}_1$, $\mathcal{V}_2$, $\mathcal{E}$} & Bipartite graph with node and edge sets.\\
    \hline
        {$\emb{A}$}, $\emb{D}$    & {Adjacency and diagonal degree matrices.}  \\
    \hline
         \tabincell{l}{$\emb{Y}$}   & \tabincell{l}{{Edge transactions where 1 indicates the} \\ interaction existence, and 0 otherwise.} \\
    \hline
        {$\widehat{\emb{Y}}$} & {Estimated matching scores.} \\
    \hline
      {$d$}  & Hash code dimension.\\
    \hline
      {$L$} &  {Numbers of convolutional hashing.}\\
    \hline
        {$\emb{V}_x^{(l)}$}  & {Node $x$'s intermediate feature at iteration $l$.} \\
    \hline
        {$\emb{Q}_x^{(l)}$}   &  {Binary embedding of node $x$ at iteration $l$.} \\
    \hline
        $\alpha^{(l)}$  & $x$'s rescaling factor computed at iteration $l$.\\
    \hline
    ${\emb\epsilon^{(l)}_x}'$, ${\emb\epsilon^{(l)}_x}''$, ${\varrho_{x}^{(l)}}$, ${\varrho_{x}^{(l)}}''$ &  Perturbation noises. \\ 
    \hline
      {${\emb{V}_x^{(l)}}'$}, {${\emb{V}_x^{(l)}}''$}   & Augmented continuous features of $x$.\\
    \hline
      {${\mathcal{Q}_x^{(l)}}'$}, {${\mathcal{Q}_x^{(l)}}''$}   & Augmented hash codes of $x$.\\
    \hline
       {${\alpha_x^{(l)}}'{\emb{Q}_x^{(l)}}$}  & Augmented binarized embeddings of $x$.\\
    \hline
        {$\mathcal{Q}_x$}  &   {Final output hash codes of node $x$.} \\
    \hline
        $\mathcal{L}_{cl}^1$, $\mathcal{L}_{cl}^2$ & {Dual contrastive loss terms.} \\ 
    \hline
        $\mathcal{L}_{cl}$, $\mathcal{L}_{bpr}$, $\mathcal{L}$ & {Two loss terms of final objective $\mathcal{L}$.} \\
    \hline
      {$\eta$, $\tau$, $\sigma$, $H$, $n$, $\lambda_1$, $\lambda_2$}  & Hyperparameters.\\
    \hline
  \end{tabular}
\end{table}

\section{{\model: Methodology}}
\label{sec:method}
\begin{figure*}[tp]
\begin{minipage}{1\textwidth}
\includegraphics[width=7.1in]{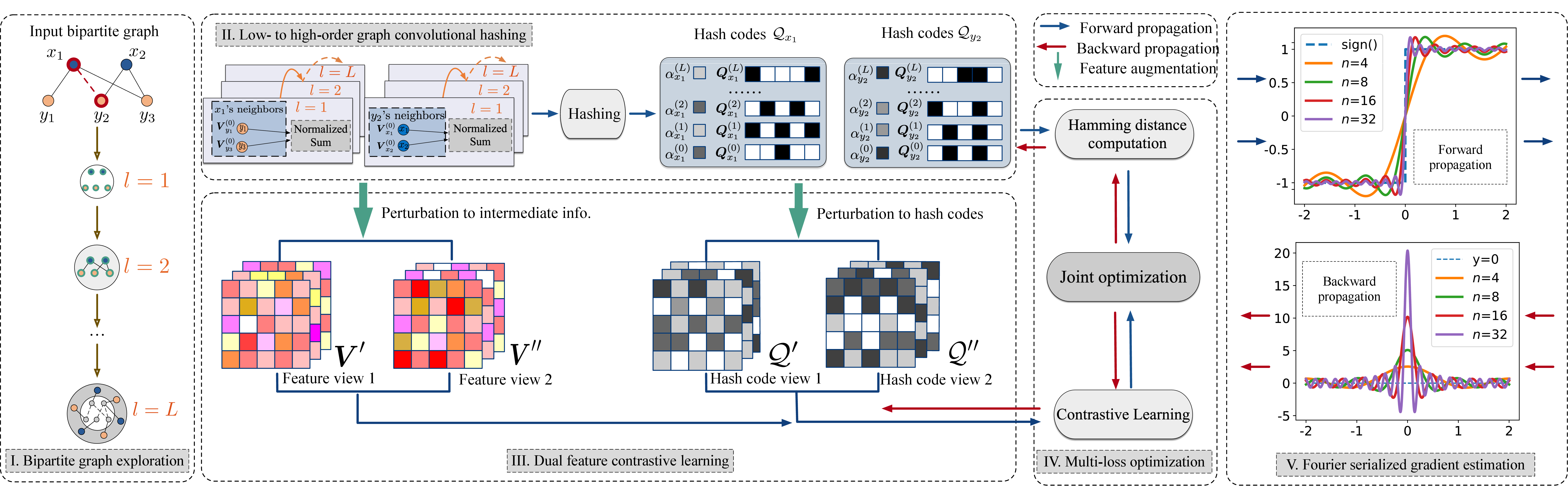}
\end{minipage} 
\caption{\newfig Workflow illustration of \model~framework (best view in color).}
\label{fig:model}
\end{figure*}

\subsection{\textbf{Overview}}
We hereby formally introduce our \model~model, which incorporates the following key modules:
(1) \textit{adaptive graph convolutional hashing} (\cref{sec:hashing}) provides an effective encoding approach to enhance the expressivity of hashed features significantly while ensuring efficient matching in Hamming space;
(2) \textit{dual feature augmentation for contrastive learning} (\cref{sec:cl}) constructs effective data augmentations for both intermediate continuous information and target hash codes in dual manner, striving to improve the quality and discriminability of the ouputs.
(3) \textit{Fourier serialized gradient estimation} (\cref{sec:ge}) introduces the Fourier Series decomposition for $\sign(\cdot)$ in the frequency domain. By leveraging this technique, more accurate gradient approximation can be achieved.
Incorporating the learned hash codes, \model~also incorporates efficient online matching using the Hamming distance measurement (\cref{sec:score}). 
Our model illustration is attached in Figure~\ref{fig:model}.

\subsection{\textbf{Adaptive Graph Convolutional Hashing}}
\label{sec:hashing}

To enhance expressivity and smooth loss landscapes, one approach is to incorporate the \textit{relaxation strategy}.
Apart from the topology-aware embedding hashing with $\sign(\cdot)$:
\begin{sequation}
\label{eq:hashing}
\emb{Q}_{x}^{(l)} = \sign({\emb{V}}_{x}^{(l)}),
\end{sequation}%
we introduce an effective way to increase the flexibility and smoothness of the learned representations, via additionally computing a layer-wise positive rescaling factor for each node as follows, such that $\alpha_x^{(l)} \in \mathbb{R}^+$ and ${\emb{V}}^{(l)}_x \approx$ $\alpha_x^{(l)} \emb{Q}^{(l)}_x$:
\begin{sequation}
\label{eq:rescale}
\alpha_x^{(l)} = \frac{1}{d} ||{\emb{{V}}}_x^{(l)}||_1.
\end{sequation}%
Instead of treating these factors as learnable parameters, such deterministic computation substantially prunes the parameter search space while attaining the adaptive approximation functionality for different graph nodes. 
We demonstrate this in~\cref{sec:ablation} of experiments.

After $L$ iterations of hashing, we obtain the table of \textbf{adaptive hash codes} $\mathcal{Q} = \{\emb{\alpha}, \emb{Q}\}$, where $\emb{\alpha} \in \mathbb{R}^{(|\mathcal{V}_1|+|\mathcal{V}_2|)\times (L+1)}$ and $\emb{Q} \in \mathbb{R}^{(|\mathcal{V}_1|+|\mathcal{V}_2|)\times (L+1) \times d}$.
For each node $x$, its corresponding hash codes are indexed and assembled:
\begin{sequation}
\resizebox{1\linewidth}{!}{$
\displaystyle
\emb{\alpha}_x = \alpha_x^{(0)} || \alpha_x^{(1)} || \cdots || \alpha_x^{(L)}, \text{ and } \emb{{Q}}_x = \emb{Q}_x^{(0)} || \emb{Q}_x^{(1)} || \cdots || \emb{Q}_x^{(L)}.
$}
\end{sequation}%
Intuitively, table $\mathcal{Q}$ encodes bipartite structural information that is propagated back and forth at different iteration steps $l$, i.e., from $0$ to the maximum step $L$.
It not only tracks the intermediate knowledge hashed for all graph nodes, but also maintains the value approximation to their original continuous embeddings, e.g., {\footnotesize ${\emb{{V}}}_x^{(l)}$}.
In addition, with the slightly more space cost (complexity analysis in~\cref{sec:complexity}), such detached hash encoding approach still supports the bitwise operations (\cref{sec:score}) for accelerating inference and matching.

{\cyk
\subsection{\textbf{Dual Feature Contrastive Learning}}
\label{sec:cl}
\subsubsection{\textbf{Dual Feature Augmentation}}
Diverging from conventional methods that manipulate graph structures, such as dropout of edges and nodes~\cite{wu2021self, you2020graph, zhu2021graph}, which often prove to be intractable and time-consuming, our research focuses on exploring contrastive learning directly on features within the embedding spaces.
Specifically, we introduce a procedure involving the addition of random noises~\cite{yu2022graph,yu2023xsimgcl} to both the embeddings before and after hashing.
Given the node $x$ with the segments of its \textit{continuous embedding $\emb{V}^{(l)}_x$} and \textit{binarized hash code $\alpha_{x}^{(l)}\emb{Q}^{(l)}_x$} at the $l$-th layer of graph convolution, we achieve the feature-level augmentations as follows.
Firstly, we conduct the augmentation for continuous embedding $\emb{V}^{(l)}_x$ as:
\begin{sequation}
\label{eq:cl}
{\emb{V}^{(l)}_x}'=\emb{V}^{(l)}_x + {\emb\epsilon^{(l)}_x}', \quad {\emb{V}_x^{(l)}}'' = {\emb{V}_x^{(l)}} + {\emb\epsilon^{(l)}_x}''.
\end{sequation}%
These perturbation vectors ${\emb\epsilon^{(l)}_x}'$, ${\emb\epsilon^{(l)}_x}''$$\in$$\mathbb{R}^d$ are drawn by: 
\begin{sequation}
\label{eq:constrain}
 \|\emb\epsilon^{(l)}\|_2 = \tau, \quad \emb\epsilon^{(l)} = \emb{\overline{\epsilon^{(l)}}} \odot \emb{Q}_x^{(l)},
\end{sequation}%
where $\emb{\overline{\epsilon^{(l)}}}$ $\in$ $\mathbb{R}^d \sim U(0,1)$ and $\odot$ is the Hadamard product.
The hyperparameter $\tau$ controls the \textit{embedding uniformity} which will be empirically analyzed later in \cref{sec:exp_cl}.
In Equation~(\ref{eq:constrain}), the first constraint mainly shapes the magnitude of perturbation vector, e.g., ${\emb\epsilon^{(l)}_x}'$, and it numerically corresponds to the point on a hypersphere with the radius $\tau$.
The second constraint ensures that $\emb{V}_x^{(l)}$, ${\emb\epsilon^{(l)}_x}'$, and ${\emb\epsilon^{(l)}_x}''$ are located in the same region.
This requirement constrains them to have the same direction with $\emb{Q}_x^{(l)}$ and prevents the addition of noises from causing significant deviations in $\emb{V}_x^{(l)}$. We visually depict the process in Figure~\ref{fig:cl}(a).

\begin{figure}[tp]
\begin{minipage}{0.5\textwidth}
\includegraphics[width=3.4in]{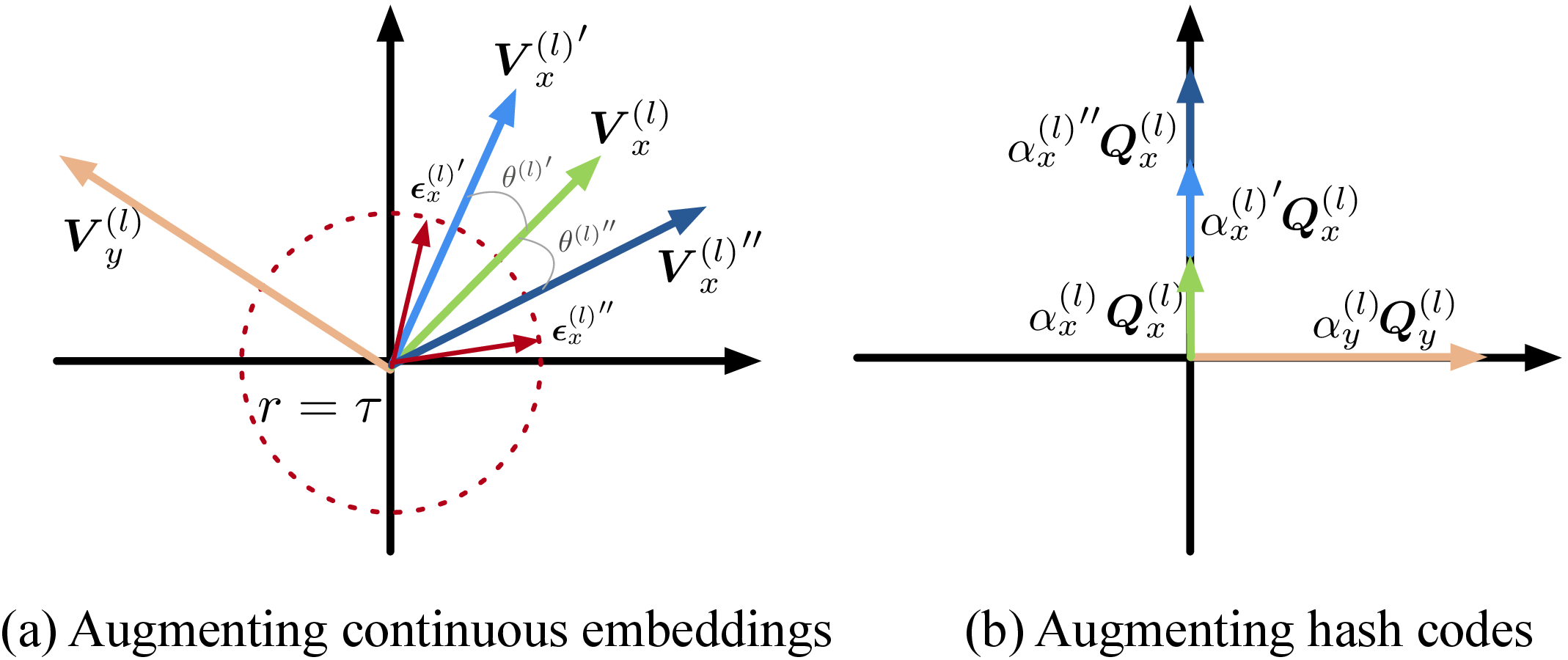}
\end{minipage} 
\caption{\newfig Dual feature augmentation for contrastive learning.}
\label{fig:cl}
\end{figure}

By incorporating these scaled noise vectors to $\emb{V}_x^{(l)}$, we essentially rotate them by two small angles, denoted as ${\theta^{(l)}}'$ and ${\theta^{(l)}}''$. 
Each rotation corresponds to a modification in $\emb{V}_x^{(l)}$ and produces a pair of augmented representations, i.e., ${\emb{V}^{(l)}_x}'$ and ${\emb{V}_x^{(l)}} ''$.
Due to the small rotation magnitude, the augmented embeddings capture both the essential features of the original information and introduce slight and acceptable variations.

Likewise, we implement the feature augmentation on the output hash codes. 
However, due to the discreteness of Hamming space, it would introduce significant variations to the binary embeddings, i.e., $\emb{Q}^{(l)}_x$, if the noise is simply perturbed on it. 
Therefore, as illustrated in Figure~\ref{fig:cl}(b), we opt to add noise to the scalar value $\alpha_{x}^{(l)}$ as follows:
\begin{sequation}
\label{eq:scalercl}
{\alpha_{x}^{(l)}}' = \alpha_{x}^{(l)} + {\varrho_{x}^{(l)}}', \quad {\alpha_{x}^{(l)}}'' = \alpha_{x}^{(l)} + {\varrho_{x}^{(l)}}'',
\end{sequation}%
where $\varrho_{x}^{(l)}$ $\in$ $\mathbb{R}$ $\sim$ $U(0,1)$.
Instead of directly perturbing $\emb{Q}^{(l)}_x$, this approach allows us to introduce controlled perturbations without affecting the binarization nature of hashing.

\subsubsection{\textbf{Dual Feature Contrastive Learning Objectives}}
After iterative graph convolutions, we achieve two views of feature-based augmentations.
For continuous intermediate information, we have the contrastive optimization term as:
\begin{sequation}
\label{eq:cl_loss1}
\mathcal{L}_{cl}^1 =\sum_{x \in \mathcal{B}}-\log \frac{\exp ( \frac{{\emb{V}_x'}^{\mathsf{T}} {\emb{V}_x''}}{\sigma}) }{\sum_{y  \in \mathcal{B}} \exp (\frac{{\emb{V}_x'}^{\mathsf{T}} {\emb{V}_y''}}{\sigma})},
\end{sequation}%
where $\mathcal{B}$ denotes a training batch and the hyperparameter $\sigma > 0$. 
${\emb{V}_x'}$ and ${\emb{V}_x''}$ are concatenated representations from:
\begin{sequation}
\resizebox{1\linewidth}{!}{$
\displaystyle
{\emb{V}_x}'={\emb{V}_x^{(0)}}'\|{\emb{V}_x^{(1)}}'\| \cdots \| {\emb{V}_x^{(L)}}', {\emb{V}_x}''={\emb{V}_x^{(0)}}''\|{\emb{V}_x^{(1)}}''\| \cdots \| {\emb{V}_x^{(L)}}''.
$}
\end{sequation}%
On the other hand, for the hash code $\mathcal{Q}_x$, we can similarly obtain two augmented views $\mathcal{Q}_x'$ and $\mathcal{Q}_x''$ and have: 
\begin{sequation}
\label{eq:cl_loss2}
\mathcal{L}_{cl}^2 =\sum_{x \in \mathcal{B}}-\log \frac{\exp(\frac{\alpha_x'\alpha_x''}{\sigma} \|{\emb{Q}_x}\|^2 )}{\sum_{y  \in \mathcal{B}} \exp (\frac{\alpha_x'\alpha_y''}{\sigma} \emb{Q}_x^{\mathsf{T}} {\emb{Q}_y})},
\end{sequation}%
where $\|{\emb{Q}_x}\|^2$ and $\emb{Q}_x^{\mathsf{T}} {\emb{Q}_y}$ can be efficiently computed in Hammming space (introduced later in~\cref{sec:score}).

Generally, these two contrastive loss terms encourage consistency between the augmented representations of the same node $x$, e.g., $\mathcal{Q}_x'$ and $\mathcal{Q}_x''$, while maximizing the disagreement between embeddings of different nodes, e.g., $x$ and $y$.
And our proposed dual contrastive feature objective can further strengthen the optmization effect to both \textit{continuous intermediate information} as well as the \textit{target hash codes}, eventually producing a fine-grained learning of hash codes with good ``uniformity''.
We provide a detailed empirical study to this design in \cref{sec:exp_cl}.

}

\subsection{\textbf{Fourier Serialized Gradient Estimation}}
\label{sec:ge}
To provide the accordant gradient estimation for hash function $\sign(\cdot)$, we approximate it by utilizing its Fourier Series decomposition in the frequency domain. 
Specifically, $\sign(\cdot)$ can be regarded as a special instance of the periodical Square Wave Function $t(x)$ within the interval of length $2H$, i.e., $\sign(\phi) = t(\phi)$, $|\phi| < H$.  
By decomposing $t(x)$ into its Fourier Series form, we shall have: 
\begin{sequation}
\sign(\phi) = \frac{4}{\pi}\sum_{i=1,3,5,\cdots}^{+\infty}\frac{1}{i}\sin(\frac{\pi i\phi}{H}), {\rm \ \ where \ \ } |\phi| < H.
\end{sequation}%

Fourier Series decomposition of $\sign(\cdot)$ with infinite terms is a lossless transformation~\cite{rust2013convergence}.
Therefore, as depicted in Figure~\ref{fig:model}(c), we can set the finite expanding term $n$ to obtain its approximation version as follows: 
\begin{sequation}
{\sign(\phi)} \doteq \frac{4}{\pi}\sum_{i=1,3,5,\cdots}^{n}\frac{1}{i}\sin(\frac{\pi i\phi}{H}).  \\
\end{sequation}%
The corresponding derivatives can be obtained as:
\begin{sequation}
\label{eq:gradient}
\frac{\partial{{\sign(\phi)}}}{\partial \phi}   \doteq \frac{4}{H} \sum_{i=1,3,5,\cdots}^{n} \cos(\frac{\pi i\phi}{H}). 
\end{sequation}%

Unlike other gradient estimators such as tanh-alike~\cite{gong2019differentiable, qin2020forward} and SignSwish~\cite{darabi2018bnn}, approximating the $\sign(\cdot)$ function with its Fourier Series does not distort the main direction of the true gradients during model optimization~\cite{xu2021learning}. 
This characteristic is advantageous as it facilitates a coordinated transformation from continuous values to their corresponding binarizations for node representations.
Consequently, it effectively preserves the discriminability of hash codes and leads to improved retrieval accuracy.
We present this performance comparison in~\cref{sec:fs_exp} of experiments. 
In summary, as indicated in Equation~(\ref{eq:formal_grad}), we utilize the strict $\sign(\cdot)$ during forward propagation to encode hashing embeddings, while estimating the gradients $\frac{\partial\sign(\phi)}{\partial \phi}$ for backward propagation.
\begin{sequation}
\label{eq:formal_grad}
\left\{ 
\begin{aligned}
& \boldsymbol{Q}^{(l)} = \sign(\phi),  &\text{Forward propagation.} \\
& \frac{\partial \boldsymbol{Q}^{(l)}}{\partial \phi} \doteq \frac{4}{H} \sum_{i=1,3,5,\cdots}^{n} \cos(\frac{\pi i\phi}{H}). & \text{Backward propagation.}
\end{aligned}
\right.
\end{sequation}%

\subsection{\textbf{Model Prediction and Optimization}}
\label{sec:score}

\subsubsection{\textbf{Matching score prediction}}
\label{sec:score_computation}
Given two nodes $x \in \mathcal{V}_1$ and $y \in \mathcal{V}_2$, one natural manner to implement the score function is \textit{inner-product}, mainly for its simplicity as:
\begin{sequation}
\label{eq:inner_score}
\widehat{\emb{Y}}_{x,y} =  (\alpha_x\emb{Q}_x)^\mathsf{T} \cdot (\alpha_y\emb{Q}_y).
\end{sequation}%
However, the inner product in Equation~(\ref{eq:inner_score}) is still conducted in the (continuous) Euclidean space with \textit{full-precision arithmetics}.
To bridge the connection between the inner product and Hamming distance measurement, we introduce the theorem:

\begin{thm}[\textbf{Hamming Distance Matching}]
\label{tm:equal}
Given two hash codes, we have $(\alpha_x\emb{Q}_x)^\mathsf{T} \cdot (\alpha_y\emb{Q}_y)$ $=$ $\alpha_x\alpha_y$ $(d - 2D_{H}(\emb{Q}_x, \emb{Q}_y))$.
\end{thm}
\begin{proof}
\begin{sequation}
\begin{aligned}
&\emb{Q}_x^\mathsf{T} \cdot \emb{Q}_y \\ 
&= \big|\{ i|(\emb{Q}_{x})_i = (\emb{Q}_{y})_i = 1\}\big| +  \big|\{ i|(\emb{Q}_{x})_i = (\emb{Q}_{y})_i = -1\}\big| \\ 
& -  \big|\{ i|(\emb{Q}_{x})_i \neq (\emb{Q}_{y})_i\}\big|\\
& = d - 2 \cdot \big|\{ i|(\emb{Q}_{x})_i \neq (\emb{Q}_{y})_i\}\big|  = \underline{d - 2D_{H}(\emb{Q}_x, \emb{Q}_y))},\\
\end{aligned}
\end{sequation}%
which completes the proof.
\end{proof}

$D_H(\cdot, \cdot)$ is Hamming distance between two inputs.
By applying Theorem~\ref{tm:equal}, we transform the score computation to the Hamming distance matching. 
This transformation allows to significantly reduce the number of floating-point operations (\#FLOPs) in the original score computation formulation (Equation~(\ref{eq:inner_score})) to efficient Hamming distance matching.
Consequently, this can develop substantial computation acceleration, as further analyzed in~\cref{sec:complexity}.

\subsubsection{\textbf{Multi-loss Objective Function}}
Our objective function comprises two components, i.e., BPR loss $\mathcal{L}_{bpr}$ and contrastive learning loss $\mathcal{L}_{cl}$. 
Generally, these two loss functions harness the regularization effect to each other.
The rationale behind this design is: 
\begin{itemize}[leftmargin=*]
\item $\mathcal{L}_{bpr}$ ranks the matching scores computed from the hash codes of a given pair of graph nodes.

\item {\cyk$\mathcal{L}_{cl}$ measures the consistency of nodes under different augmentation views represented by both continuous and binarized hash codes.}
\end{itemize}
Concretely, we implement $\mathcal{L}_{bpr}$ with \textit{Bayesian Personalized Ranking} (BPR) loss as follows:  
\begin{sequation}
\label{eq:hd-bpr}
\mathcal{L}_{bpr} = -\sum_{x \in \mathcal{V}_1} \sum_{\tiny y\in \mathcal{N}(x) \atop y'\notin \mathcal{N}(x)} \ln \sigma(\widehat{\emb{Y}}_{x,y} - \widehat{\emb{Y}}_{x,y'}).
\end{sequation}%
$\mathcal{L}_{bpr}$ encourages the predicted score of an observed edge to be higher than its unobserved counterparts~\cite{lightgcn}.
{\cyk As for $\mathcal{L}_{cl}$, we take summation of both loss terms, i.e., $\mathcal{L}_{cl}^1$ and $\mathcal{L}_{cl}^2$ introduced in \cref{sec:cl}, for dual feature contrastive learning:
\begin{sequation}
\mathcal{L}_{cl} = \mathcal{L}_{cl}^1 + \mathcal{L}_{cl}^2.
\end{sequation}%
Let $\lambda_1$ be a hyperparameter and $\Theta$ denote the set of trainable embeddings regularized by the parameter $\lambda_2$ to avoid over-fitting.
Our final objective function is defined as:
\begin{sequation}
\label{eq:L}
\mathcal{L} = \mathcal{L}_{bpr} + \lambda_1\mathcal{L}_{cl} + \lambda_2 ||\Theta||_2^2.
\end{sequation}%
While $\mathcal{L}_{bpr}$ focuses on pairwise preferences and rankings, the dual contrastive learning loss helps in capturing fine-grained similarities and differences between node representations. 
By incorporating these loss terms, the model not only prioritizes the learning of ranking information embedded in the outputs, but also produces high-quality hash codes with better performances.
The ablation study of such multi-loss optimization design is conducted in~\cref{sec:ablation}.
}

So far, we have introduced all technical parts of \model~and attached the pseudocodes in Algorithm~\ref{alg:model}. 
To better understand the model scalability of \model, we provide the complexity analysis in the following section.

\begin{algorithm}[t]
\small
{\newfig
\caption{\model~training algorithm.}
\label{alg:model}
\LinesNumbered  
\While{\rm{model does not converge}}{
	 \For{$l = 0, \cdots, L-1$}{
          $\emb{{Q}}^{(l+1)} \gets \sign\big({\emb{V}}^{(l+1)})$ \Comment*[r]{Eq.(\ref{eq:hashing})}
          $\emb{\alpha} \gets$ calculate the rescaling factors \Comment*[r]{Eq.(\ref{eq:rescale})}

           \For{$x \in \mathcal{V}_1, y \in \mathcal{N}(x)$}{
              
              ${\emb{V}_x^{(l+1)}}'$, ${\emb{V}_x^{(l+1)}}''$, ${\emb{V}_y^{(l+1)}}''$, \\ 
              ${\alpha_x^{(l+1)}}$, ${\alpha_x^{(l+1)}}''$, ${\alpha_y^{(l+1)}}''$  $\gets$ construct dual feature augmentation \Comment*[r]{Eq.(\ref{eq:cl}-\ref{eq:scalercl})}
               $\mathcal{L}_{cl}^1$, $\mathcal{L}_{cl}^2$ $\gets$ contrastive loss \Comment*[r]{Eq.(\ref{eq:cl_loss1}-\ref{eq:cl_loss2})}
              $\widehat{\emb{Y}}_{x,y} \gets$ {\scriptsize $\alpha_x\alpha_y$$(d - 2D_{H}(\emb{Q}_x, \emb{Q}_y))$} \Comment*[r]{Eq.(\ref{eq:inner_score})}
              $\mathcal{L}_{bpr}$ $\gets$ BPR ranking loss  \Comment*[r]{Eq.(\ref{eq:hd-bpr})}
            }

        }
     
      $\mathcal{L} \gets$ accumulate loss for optimization \Comment*[r]{Eq.(\ref{eq:L})} 
      
}
\textbf{Function} \tt{Gradient\_estimator}($\mathcal{L}$): \\
$\frac{\partial \mathcal{L}}{\partial \emb{V}} \gets \frac{\partial \mathcal{L}}{\partial \emb{Q}} \cdot \frac{4}{H} \sum_{i=1,3,5,\cdots}^{n} \cos(\frac{\pi i \emb{V}}{H})$ \Comment*[r]{Eq.(\ref{eq:gradient})}
}
\end{algorithm}

{\cyk
\section{{Complexity Analysis}}
\label{sec:complexity}
\textbf{Training time complexity.}
$|\mathcal{E}|$, $|\mathcal{B}|$, $s$, and $n$ are the edge number, batch size, numbers of train iterations and Fourier Series decomposition terms, respectively.
As shown in Table~\ref{tab:train_time}, we derive that:
(1) The time complexity of the graph normalization to the original adjacency matrix, i.e., $\emb{D}^{-\frac{1}{2}} \emb{A} \emb{D}^{-\frac{1}{2}}$, is $O(2|\mathcal{E}|)$.
(2) In graph convolution and hashing, the time complexity is $O(\frac{2dsL|\mathcal{E}|^2}{|\mathcal{B}|})$, where $L \leq 4$ is a common setting~\cite{lightgcn,ngcf,kipf2016semi,graphsage} to avoid \textit{over-smoothing}~\cite{li2019deepgcns}.
(3) \model~takes $O\big(2sd|\mathcal{E}|\big)$ to compute $\mathcal{L}_{bpr}$ loss.
As for $\mathcal{L}_{cl}$ loss, for each node $x$, we need to conduct random noise addition to all other nodes in $\mathcal{B}$, as shown in Equations~(\ref{eq:cl_loss1}-\ref{eq:cl_loss2}), which is known as the widely-used in-batch sampling~\cite{chen2020simple}.
Thus the total loss complexity for our dual feature contrastive learning would be {$O\big(2|\mathcal{B}|sd|\mathcal{E}|\big)$}.
(4) Lastly, \model~takes $O(snd|\mathcal{E}|)$ to estimate the gradients for the $d$-dimension hash codes.
Thus, the total complexity in total is quadratic to the graph edge numbers, i.e., $|\mathcal{E}|$.
This is common in GCN frameworks and actually acceptable for large bipartite graphs, which may dispel the concerns of large training cost in practice.

 \begin{table}[t]
\newfig
\centering
\scriptsize
\caption{\newfig Traing time complexity.}
\label{tab:train_time}
\setlength{\tabcolsep}{1.5mm}{
  \begin{tabular}{c | c | c | c | c }
\toprule
    {Graph Norm.}   & {Conv. \& Hash.}   & {BPR Loss} &{CL loss}  &{Grad. Est.} \\
\midrule[0.1pt]
    {$O(2|\mathcal{E}|)$}  & {$O(\frac{2dsL|\mathcal{E}|^2}{|\mathcal{B}|})$}  &{$O\big(2sd|\mathcal{E}|\big)$} & {$O\big(2|\mathcal{B}|sd|\mathcal{E}|\big)$} & { $O(snd|\mathcal{E}|)$} \\
\bottomrule
\end{tabular}}
\end{table}

\begin{table}[t]
\centering
\newfig
\footnotesize
\caption{\newfig Complexity of space and computation.}
\label{tab:prediction}
\setlength{\tabcolsep}{1mm}{
\begin{tabular}{c | c | c c}
\toprule
  ~          & {\footnotesize Space cost}   &  {\footnotesize \#FLOP} & {\footnotesize \#BOP}       \\
\midrule
 {\scriptsize float32-based}      & {\scriptsize $32|\mathcal{V}_1 \cup \mathcal{V}_2|d$}       &   {\scriptsize $O\big(|\mathcal{V}_1| \cdot |\mathcal{V}_2| d\big)$}      &   {-}         \\
\midrule[0.1pt]
{\scriptsize \model}       & {\scriptsize $|\mathcal{V}_1 \cup \mathcal{V}_2| (d+32(L+1))$}    & {\scriptsize $O\big(4|\mathcal{V}_1| \cdot |\mathcal{V}_2| \big)$}            & {\scriptsize $O\big(|\mathcal{V}_1| \cdot |\mathcal{V}_2| d\big)$}    \\
\bottomrule
\end{tabular}}
\end{table}


\textbf{Hash codes space cost.}
As shown in Table~\ref{tab:prediction}, the total space cost of hash codes is {\small $O(|\mathcal{V}_1 \cup \mathcal{V}_2| (d+32(L+1)))$} bits, supposing that we use float32 for those rescaling factors in $L+1$ iterations.
Compared to the continuous embedding size, i.e., $32|\mathcal{V}_1 \cup \mathcal{V}_2|d$ bits, the size reduction ratio thus is:
\begin{sequation}
\label{eq:space}
\begin{aligned}
ratio = \frac{32|\mathcal{V}_1 \cup \mathcal{V}_2|d}{|\mathcal{V}_1 \cup \mathcal{V}_2| (d+32(L+1))}  & = \frac{32d}{d+32(L+1)} \\
& = \frac{32}{1 + 32(\frac{L}{d} + \frac{1}{d})}.
\end{aligned}
\end{sequation}%
Based on our previous explanation, it has been noted that excessively stacking iteration layers can lead to a decline in performance~\cite{li2019deepgcns}. Therefore, if $L$ $\leq$ $4$, \model~has the potential to achieve a significant increase in the size compression ratio, if the embedding dimension is increased.

{\textbf{Online matching.}}
The original score formulation in Equation~(\ref{eq:inner_score}) contains $d$ floating-point operations (\#FLOPs).
As shown in Table~\ref{tab:prediction}, using Hamming distance matching can convert the most of floating-point arithmetics to binary operations (\#BOPs), with a few \#FLOPs for scalar computations, i.e., $4\ll d$.

}

\section{{Experimental Evaluation}}
\label{sec:exp}
We evaluate \model~with the aim of answering the following research questions:
\begin{itemize}[leftmargin=*]
\item \textbf{RQ1.} How does \model~compare to the state-of-the-art hashing-based models in Top-N Hamming retrieval?

\item \textbf{RQ2.} What is the performance gap between \model~and full-precision models \textit{w.r.t.} long-list retrieval quality?

\item {\cyk \textbf{RQ3.} How does our proposed dual feature contrastive learning contribute to the performance of \model?}

\item \textbf{RQ4.} What are advantages of other model components?

\item \textbf{RQ5.} What is the resource consumption of \model?

\item \textbf{RQ6.} How does Fourier Series decomposition perform in terms of retrieval accuracy and training efficiency?
\end{itemize}

\subsection{\textbf{Experiment Setup}}
\label{sec:exp_setup}

\textbf{Datasets and Evaluation Metrics.} 
We include six real-world bipartite graphs in Table~\ref{tab:datasets} that are widely evaluated~\cite{lightgcn,chen2021hyper,chen2021attentive,yang2022hrcf,ngcf,zhang2022knowledge} as follows:
\begin{enumerate}[leftmargin=*]
\item \textbf{MovieLens}\footnote{\notsotiny\url{https://grouplens.org/datasets/movielens/1m/}} is a widely adopted benchmark for movie review. Similar to the setting in~\cite{hashgnn}, if the user $x$ has rated item $y$, we set the edge $\emb{Y}_{x,y} = 1$, 0 otherwise. 
\item \textbf{Gowalla}\footnote{\notsotiny\url{https://github.com/gusye1234/LightGCN-PyTorch/tree/master/data/gowalla}}~\cite{ngcf,hashgnn,lightgcn,dgcf} is the dataset~\cite{liang2016modeling} between \textit{customers} and \textit{check-in locations} collected from Gowalla. 
\item \textbf{Pinterest}\footnote{\notsotiny\url{https://sites.google.com/site/xueatalphabeta/dataset-1/pinterest_iccv}} is an open dataset for image recommendation between \textit{users} and \textit{images}.
Edges represent the pins over images initiated by users. 
\item \textbf{Yelp2018}\footnote{\notsotiny\url{https://github.com/gusye1234/LightGCN-PyTorch/tree/master/data/yelp2018}} is from Yelp Challenge 2018 Edition, bipartitely modeling between \textit{users} and \textit{local businesses}.
\item \textbf{AMZ-Book}\footnote{\notsotiny\url{https://github.com/gusye1234/LightGCN-PyTorch/tree/master/data/amazon-book}} is the bipartite graph between \textit{readers} and \textit{books}, organized from Amazon-review~\cite{he2016ups}.  
\item \textbf{Dianping}\footnote{\notsotiny\url{https://www.dianping.com/}} is a commercial dataset between \textit{users} and \textit{local businesses} recording their diverse interactions, e.g., clicking, saving, and purchasing. 
\end{enumerate}

\begin{table}[t]
\centering
\footnotesize
\caption{The statistics of datasets.}
\label{tab:datasets}
\setlength{\tabcolsep}{0.2mm}{
\begin{tabular}{c | c | c | c | c | c | c}
\toprule 
             & {\footnotesize MovieLens}  & {\footnotesize Gowalla}   & {\footnotesize Pinterest}  &  {\footnotesize Yelp2018} & {\footnotesize AMZ-Book} & {\footnotesize Dianping}\\
\midrule[0.1pt]
    {\footnotesize $|\mathcal{V}_1|$ }  & {6,040}   & {29,858}   & {55,186}   & {31,668}  &{52,643}  &{332,295}  \\ 
    {\footnotesize $|\mathcal{V}_2|$ }  & {3,952}   & {40,981}   & {9,916}    & {38,048}  &{91,599}  &{1,362}  \\
\midrule[0.1pt]
    {\footnotesize $|\mathcal{E}|$ } & {1,000,209} & {1,027,370} & {1,463,556} & {1,561,406} & {2,984,108} &{10,000,014} \\
   \midrule[0.1pt]
 Density  & {0.04190}   & {0.00084}   & {0.00267}   & {0.00130}  &{0.00062}  &{0.02210}  \\
\bottomrule
\end{tabular}}
\end{table}

\begin{table}[tp]
\centering
\footnotesize
\newfig
\caption{\newfig Hyperparameter settings.}
\label{tab:hyperparameter}
\setlength{\tabcolsep}{0.9mm}{
\begin{tabular}{c | c | c | c | c | c | c }
\toprule
            & MovieLens & Gowalla & Pinterest & Yelp2018 & AMZ-Book & Dianping \\
\midrule[0.1pt]
  $\eta$            &{$1\times 10^{-2}$}  &{$1\times 10^{-3}$}  &{$1\times 10^{-3}$}  &{$1\times 10^{-3}$}  &{$1\times 10^{-3}$}  &{$1\times 10^{-4}$}  \\
  $\lambda_1$     &{$1\times 10^{-2}$}  &{$5\times 10^{-2}$}  &{$1\times 10^{-2}$}  &{$1\times 10^{-2}$}  &{$1\times 10^{-2}$}  &{$5\times 10^{-3}$}  \\
  $\lambda_2$     &{$1\times 10^{-5}$}  &{$1\times 10^{-5}$}  &{$1\times 10^{-4}$}  &{$1\times 10^{-4}$}  &{$1\times 10^{-5}$}  &{$1\times 10^{-5}$}  \\
  $\tau$      &{0.1}  &{0.1}  &{0.1}  &{0.1}  &{0.1}  &{0.1}  \\
  $\sigma$      &{0.2}  &{0.2}  &{0.2}  &{0.2}  &{0.1}  &{0.2}  \\
  $n$     &{8}  &{16}  &{8}  &{4}  &{8}  &{4}  \\
  $H$     &{1}  &{1}  &{1}  &{1}  &{1}  &{1}  \\
\bottomrule
\end{tabular}
}
\end{table}

\textbf{Evaluation Metrics.}
To assess the model performance in Hamming space retrieval over bipartite graphs, we use the hash codes to find the Top-N answers for a given query node based on the closest Hamming distances. We then evaluate the ranking capability using two commonly used evaluation protocols, i.e., Recall@N and NDCG@N.

\textbf{Baselines.}
\label{sec:baseline}
In addition to BGCH, we include the following representative hashing-based models for (1) general object retrieval (LSH~\cite{lsh}), (2) image search (HashNet~\cite{hashnet}), and (3) Top-N candidate generation for recommendation (Hash\_Gumbel~\cite{gumbel1,gumbel2}, CIGAR~\cite{kang2019candidate} and HashGNN~\cite{hashgnn}).
We also include several state-of-the-art full-precision (i.e., denoted by FT32 as implemented with float32 in experiments) recommender models, i.e., NeurCF~\cite{neurcf}, NGCF~\cite{ngcf}, DGCF~\cite{dgcf}, LightGCN~\cite{lightgcn}, for the long-list ranking quality comparison:
\label{app:baselines}
\begin{enumerate}[leftmargin=*]
\item \textbf{LSH}~\cite{lsh} is a classical hashing method. LSH is proposed to approximate the similarity search for massive high-dimensional data and we introduce it for Top-N object search by following the adaptation in~\cite{hashgnn}. 

\item \textbf{HashNet}~\cite{hashnet} is a representative deep hashing method that is originally proposed for multimedia retrieval tasks.
Similar to~\cite{hashgnn}, we adapt it for graph data by modifying it with the general graph convolutional network.

\item \textbf{CIGAR}~\cite{kang2019candidate} is a state-of-the-art neural-network-based framework for fast Top-N candidate generation in recommendation. 
CIGAR can be further followed by a full-precision re-ranking algorithm. And we only use its hashing part for fair comparison.

\item \textbf{Hash\_Gumbel} is a variance of \model~with Gumbel-softmax for hash encoding and gradient estimation~\cite{gumbel1,gumbel2}.
Specifically, we first expand each embedding bit to a size-two one-hot encoding. 
Then it utilizes the Gumbel-softmax trick to replace $\sign(\cdot)$ as relaxation for binary hash code generation. 

\item \textbf{HashGNN}~\cite{hashgnn} is the state-of-the-art learning to hash based method with GCN framework. 
We use HashGNN$_{h}$ to denote the vanilla version with \textit{hard encoding} proposed in~\cite{hashgnn}, where each element of user-item embeddings is strictly binarized. 
We use HashGNN$_{s}$ to denote its proposed approximated version.

\item \textbf{BGCH}~\cite{chen2023bipartite} is the vanilla graph-base method with no feature augmentations for bipartitle graph hashing.

\item \textbf{NeurCF}~\cite{neurcf} is one representative deep neural network model for collaborative filtering in recommendation. 

\item \textbf{NGCF}~\cite{ngcf} is one of the representative graph-based models with collaborative filtering methodology. 

\item \textbf{DGCF}~\cite{dgcf} is a graph-based model that learns disentangled user intents with state-of-the-art performance. 

\item \textbf{LightGCN}~\cite{lightgcn} is another state-of-the-art GCN-based recommender model that has been widely evaluated. 

\end{enumerate}

The early hashing methods such as SH~\cite{weiss2008spectral}, RMMH~\cite{joly2011random}, and LCH~\cite{zhang2010laplacian} are not considered in our evaluation due to the established performance superiority of the competing models~\cite{hashnet,kang2019candidate} over them.

{\cyk
{\textbf{Implementations and Hyperparameter Settings.}}
We implemented our models using Python 3.6 and PyTorch 1.14.0 on a Linux machine equipped with 4 Nvidia V100 GPUs and 4 Intel Core i7-8700 CPUs.
For the baselines, we followed the official hyperparameter settings or conducted a grid search if the settings were not available.
In the following sections, we set the dimension $d$ to 256 and graph convolutions $L$ to 2.
The learning rate $\eta$ and the coefficients $\lambda_1$, $\lambda_2$ were tuned among \{$5\times10^{-5}, 10^{-5}, 5\times10^{-4}, 10^{-4}, 5\times10^{-3}, 10^{-3}, 5\times10^{-2}, 10^{-2}$\}. 
We initialize and optimize all models with default normal initializer and Adam optimizer~\cite{adam}. 
All hyperparameters are reported in Table~\ref{tab:hyperparameter}.
}

\begin{table*}[t]
\setlength{\abovecaptionskip}{0.2cm}
\setlength{\belowcaptionskip}{0.2cm}
\centering
\small
  \caption{Results of Recall@20 and NDCG@20 in Top-1000 retrieval: (1) ``R'' and ``N'' denote the Recall and NDCG; (2) the bold indicate \model~and the underline represents the best and second-best performing models; (3) \textbf{$^*$} denotes scenarios where Wilcoxon signed-rank tests indicate statistically significant improvements with over 95\% confidence level.}
  \label{tab:topn}
  \setlength{\tabcolsep}{1.5mm}{
  \begin{tabular}{c|c c| c c| c c| c c| c c|c c} 
    \toprule
    Dataset & \multicolumn{2}{c|}{MovieLens (\%)} & \multicolumn{2}{c|}{Gowalla (\%)} & \multicolumn{2}{c|}{Pinterest (\%)} & \multicolumn{2}{c|}{Yelp2018 (\%)}  & \multicolumn{2}{c|}{AMZ-Book (\%)} & \multicolumn{2}{c}{Dianping (\%)} \\
    Metric & R@20 & N@20  & R@20 & N@20  & R@20 & N@20  & R@20 & N@20  & R@20 & N@20  & R@20 & N@20    \\ 
    \midrule[0.5pt]
    LSH                 & {11.38} & {25.87}  & {8.14} & {12.23}  & {7.88} & {6.71}  & {2.91} & {4.35}  & {2.41} & {2.34}  & {5.85} & {5.84}  \\
    HashNet             & {15.43} & {32.23}  & {11.38} & {13.74}  & {10.27} & {7.33}  & {3.37} & {4.41}  & {2.86} & {2.71}  & {6.24} & {5.59}  \\
    CIGAR               & {14.84} & {31.73}  & {11.57} & {14.21}  & {10.34} & {8.53}  & {3.65} & {4.57}  & {3.05} & {3.03}  & {6.91} & {6.03}  \\
    Hash\_Gumbel            & {16.62}  & {32.48} & {12.26}  & {14.68} & {10.53}  & {8.74} & {3.85}  & {5.12} & {2.69}  & {3.24} & {8.29}  & {6.43} \\
    HashGNN$_{\rm h}$   & {14.21} & {31.83}  & {11.63} & {14.21}  & {10.15} & {8.67}  & {3.77} & {5.04}  & {3.09} & {3.15}  & {8.34} & {6.68}  \\
    HashGNN$_{\rm s}$   & {19.87} & {33.21}  & {13.45} & {14.87}  & {12.38} & {9.11}  & {4.86} & {5.34}  & {3.34} & {3.45}  & {9.57} & {7.13}  \\
    \midrule[0.1pt]
    {BGCH} &\underline{22.86} &\underline{36.26} &\underline{16.73} &\underline{16.48} &\underline{12.78} &\underline{9.42} &\underline{5.51} &\underline{5.84} &\underline{3.48} &\underline{3.92} &\underline{10.66} &\underline{7.63}  \\

    \textbf{\model} &\textbf{24.37$^*$} &\textbf{37.21$^*$} &\textbf{17.19$^*$} &\textbf{17.10$^*$} &\textbf{15.36$^*$} &\textbf{11.21$^*$} &\textbf{5.96$^*$} &\textbf{6.17$^*$} &\textbf{3.78$^*$} &\textbf{4.35$^*$} &\textbf{12.05$^*$} &\textbf{8.81$^*$}  \\
    \% Gain       &{ 6.61\%} &{2.62\%}    &{ 2.75\%} &{ 3.76\%}  &{ 20.19\%} &{ 19.00\%}  &{ 8.17\%} &{ 5.65\%}  &{ 8.62\%} &{ 10.97\%}  &{ 13.04\%} &{ 15.47\%} \\

   \bottomrule
  \end{tabular}}
\end{table*}

\begin{figure*}[t]
\begin{minipage}{1\textwidth}
\includegraphics[width=7.2in]{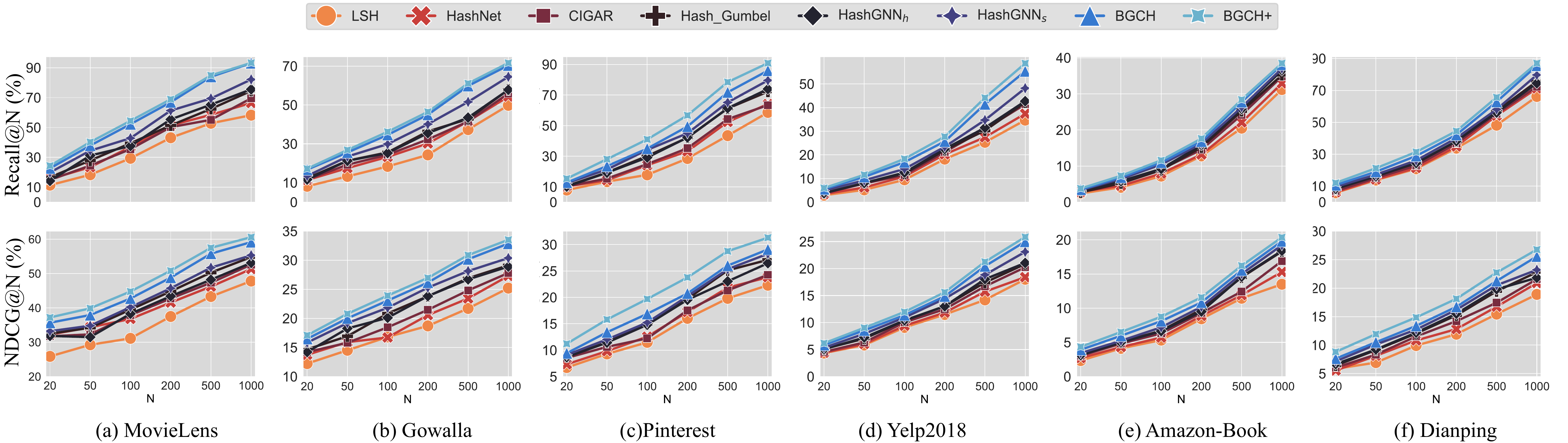}
\end{minipage} 
\caption{Top-N retrieval quality with N in \{20, 50, 100, 200, 500, 1000\} (best view in color).}
\label{fig:topn}
\end{figure*}

\subsection{\textbf{Top-N Hamming Space Query (RQ1)}}
\label{sec:exp_topn}

To assess the \textbf{fine-to-coarse} Top-N ranking capability, we fix N=1000. 
We initially present the results of Recall@20$_{1000}$ and NDCG@20$_{1000}$\footnote{We use notation Recall@20, NDCG@20 if no ambiguity is caused.} for the Top-1000 search in Table~\ref{tab:topn}. 
Additionally, we plot the holistic Recall and NDCG metric curves for the \{20, 50, 100, 200, 500, 1000\} of Top-1000 ranking in Figure~\ref{fig:topn}.
For fair comparison, we ensure that both \model~and the baselines had the convolution iteration number of 2 and the embedding dimension of 256.

\begin{itemize}[leftmargin=*]

\item \textbf{The results clearly establish the superiority of our \model~over previous hashing-based models.}
(1) Firstly, as shown in Table~\ref{tab:topn}, HashGNN outperforms traditional hashing-based baselines such as LSH, HashNet, CIGAR. 
This suggests that directly adapting conventional non-graph-based hashing methods may struggle to achieve comparable performance due to the effectiveness of \textit{graph convolutional} architecture in capturing latent information within the bipartite graph topology for hash encoding preparation.
(2) Secondly, both BGCH and \model~consistently outperform HashGNN over all datasets, thanks to the proposed \textit{adaptive graph convolutional hashing}.
\model~further achieves improvement over its vanilla version BGCH by 2.75\%$\sim$20.19\%, and 2.62\%$\sim$ 19.00\% \textit{w.r.t.} Recall@20 and NDCG@20, respectively. 
This highlights the effectiveness of our newly proposed dual feature contrastive learning paradigm. A comprehensive analysis will be conducted in~\cref{sec:exp_cl}.
(3) Thirdly, we conduct the Wilcoxon signed-rank tests on \model.
The results confirm that all improvements of \model~over BGCH are statistically significant at a 95\% confidence level.
Considering the performance improvement of \model~and BGCH over all other methods, this demonstrates the effectiveness of all proposed modules contained therein. Detailed ablation study will be introduced in~\cref{sec:ablation}.

\item \textbf{By varying N from 20 to 1000, both BGCH and \model~consistently demonstrate competitive performance compared to the baselines.}
The observations from Figure~\ref{fig:topn} are as follows:
(1) Compared to the approximated version of HashGNN, i.e., HashGNN$_{s}$, both BGCH and \model~consistently exhibit stable and significant improvements in both Recall and NDCG metrics across all six benchmarks for N ranging from 20 to 1000. 
(2) In addition to achieving higher retrieval quality, BGCH and \model~possess another advantage over HashGNN$_{s}$: they retain support for bitwise operations, specifically hamming distance matching, which facilitates accelerated inference.
However, HashGNN$_{s}$ adopts a Bernoulli random variable to determine the probability of replacing certain digits in the hash codes with original continuous values, which disables bitwise computation.
As detailed in~\cref{sec:resource}, such design achieves over 8$\times$ inference acceleration compared to HashGNN$_{s}$, which is particularly promising for query-based online matching and retrieval applications.

\end{itemize}

\subsection{\textbf{Comparing to FT32-based Models (RQ2)}}
\label{sec:exp_full}

\begin{table*}[t]
\setlength{\abovecaptionskip}{0.2cm}
\setlength{\belowcaptionskip}{0.2cm}
\centering
\small
\cykk
\caption{Results of Float32-based models.}
  \label{tab:full}
  \setlength{\tabcolsep}{1mm}{
  \begin{tabular}{c|c c| c c| c c| c c| c c|c c} 
    \toprule
    Dataset & \multicolumn{2}{c|}{MovieLens (\%)} & \multicolumn{2}{c|}{Gowalla (\%)} & \multicolumn{2}{c|}{Pinterest (\%)} & \multicolumn{2}{c|}{Yelp2018 (\%)}  & \multicolumn{2}{c|}{AMZ-Book (\%)} & \multicolumn{2}{c}{Dianping (\%)} \\
    Metric & R@1000 & N@1000  & R@1000 & N@1000  & R@1000 & N@1000  & R@1000 & N@1000  & R@1000 & N@1000  & R@1000 & N@1000    \\ 
    \midrule[0.5pt]
    NeurCF           & {96.90} & {58.76}  & {73.28} & {32.07}  & {91.29} & {28.79}  & {58.83} & {24.69}  & {40.29} & {19.83}  & {89.39} & {25.54}  \\
    NGCF             & {97.32} & {60.28}  & {76.16} & {32.13}  & {92.93} & {29.78}  & {59.97} & {25.23}  & {41.22} & {20.37}  & {90.92} & {25.76}  \\
    DGCF             & \textbf{98.48} & {62.41}  & {76.90} & {34.97}  & {96.52} & \textbf{31.47}  & {62.18} & {26.28}  & {42.71} & {21.74}  & {92.66} & {26.87}  \\
    LightGCN         & {98.27} & \textbf{62.88}  & \textbf{77.74} & \textbf{35.26}  & \textbf{96.59} & {31.32}  & {62.31} & \textbf{26.55}  & \textbf{43.89} & \textbf{21.92}  & \textbf{94.37} & \textbf{27.28}  \\
    \midrule[0.1pt]
    {BGCH} 				&{90.44} &{59.16} &{70.45} &{32.87}		&{86.30} &{29.09}&{56.11} &{25.01}&{38.27} &{19.79}&{88.26} &{25.57} \\
     \% Capability       &{91.84\%} &{94.08\%} &{90.62\%} &{93.22\%}    &{89.07\%} &{92.44\%}   &{90.05\%} &{94.20\%}   &{87.20\%} &{90.28\%}   &{93.53\%} &{93.73\%}    \\
    \midrule[0.01pt]
    \textbf{\model} 	&{92.43} &{60.60} &{72.22} &{33.52}&{91.29} &{31.31}&{58.28} &{25.85}&{39.38} &{20.36}&{89.83} &{26.77} \\
    \% Capability       &{93.86\%} &{96.37\%} &{92.90\%} &{95.07\%}   &{94.51\%} &{99.49\%}   &{93.53\%} &{97.36\%}   &{89.72\%} &{92.88\%}   &{95.19\%} &{98.13\%}    \\

   \bottomrule
  \end{tabular}}
\end{table*}

In this section, we compare both BGCH and \model~with several full-precision (FT32-based) models to evaluate the long-list search quality. 
The observations from Table~\ref{tab:full} are as follows:
(1) Both BGCH and \model~consistently deliver competitive performance compared to early full-precision models, e.g., NeurCF and NGCF, across all datasets. 
In terms of the state-of-the-art model LightGCN, both BGCH and \model~achieve over {\cykk 87\% of the Recall@1000 and 90\% of NDCG@1000 capability.
(2) We notice that, compared to NDCG@1000 metric, the Recall@1000 results are generally with larger numerical values. 
This is because the recall metric focuses on how many ground-truth items are retrieved, and thus these models perform quite well.
On the other hand, the NDCG metric cares more about the ranking capability, which usually leads to a relatively smaller values.
(3) The performance of BGCH and \model~highlights their effectiveness in ensuring high-quality long-list Top-N retrieval with competitive Recall@1000 and NDCG@1000 metrics, where \model~further improves its performance across all datasets. }
This is particularly valuable in industrial applications such as recommender systems, which involve two major stages: \textit{candidate generation} and \textit{re-ranking}. 
Having a high-quality candidate generation significantly reduces the complexity of the subsequent re-ranking stage as the search space is substantially pruned.
(4) Considering the \textit{efficiency in Hamming space retrieval} and the \textit{reduced space cost} of those learned hash codes, BGCH and \model~can be seen as viable alternatives to these full-precision models to balance retrieval performance and computational efficiency, especially in scenarios with limited computation resources.

{\cyk
\subsection{\textbf{Study of Dual Feature Contrastive Learning (RQ3)}}
\label{sec:exp_cl}

\begin{figure*}[t]
\begin{minipage}{1\textwidth}
\includegraphics[width=7.2in]{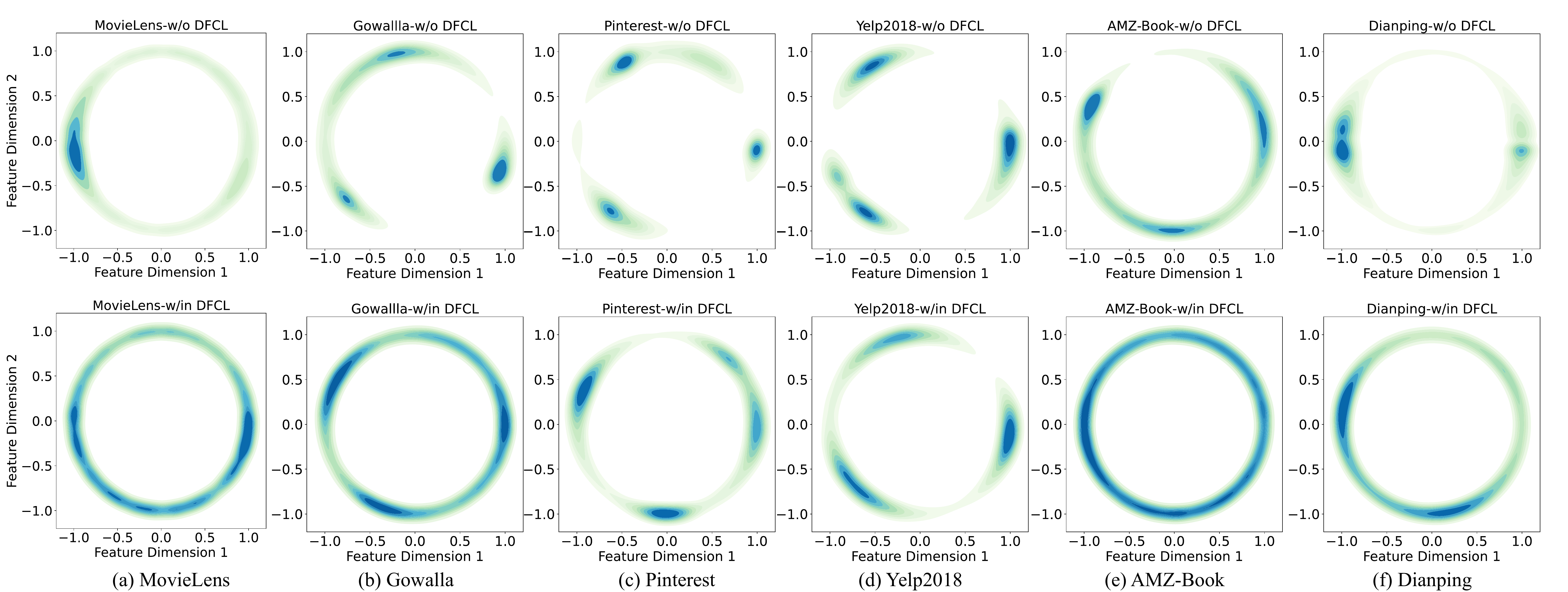}
\end{minipage} 
\caption{\newfig Distribution illustration of learned hash codes between BGCH and \model~using Gaussian kernel density estimation (KDE) over six datasets with (1) bandwidth as 0.1 and (2) number of contour levels as 10 (best view in color).}
\label{fig:kde}
\end{figure*}

In this section, we conduct a comprehensive empirical analysis to examine the impact of our dual feature contrastive learning on the quality of hashed representations.

\textbf{Structure Manipulation V.S. Feature Augmentation.}
The conventional contrastive learning usually requires explicit graph structure manipulation~\cite{liu2022graph} such as:
\begin{itemize}[leftmargin=*]
\item \textit{Node dropout (denoted by ND}): with a certain probability, the graph node and its connected edges are discarded.

\item \textit{Edge dropout (denoted by ED)}:  it drops out the edges in a graph with a dropout ratio. 

\item \textit{Graph random walk (denoted by GRW)}: it essentially operates as the multi-layer edge dropouts.
\end{itemize}

We implement these structural manipulation strategies and show the comparison results in Table~\ref{tab:othercl}.
We can clearly observe that:
(1) Edge dropout (ED) generally achieves the most competitive performance among all structure augmentation strategies.
(2) Our model \model~incorporates dual feature augmentation in the embedding space for contrastive learning, further improving performance across all datasets compared to the ED variant. This highlights the effectiveness of our proposed approach.
{\cykk (3) To provide a more fine-grained comparison, we further combine these strategies. As shown in Table~\ref{tab:othercl}, we observe that variants with ED generally perform well, compared to the other; however, even with all these strategies integrated, it still consistently under-performs our model \model.
(4) Considering the heavy time cost of all these explicit structural manipulation, our feature-wise augmentation offers more flexibility when training \model~in batch on the fly. This makes it better suited for handling larger bipartite graphs.
}

\begin{table}[t]
\centering
\notsotiny
\newfig
\caption{\newfig Comparison of structure-manipulation-based variants and \model~in terms of Recall@20.}
\label{tab:othercl}
\setlength{\tabcolsep}{0.5mm}{
\begin{tabular}{c |c | c | c | c | c| c}
\toprule
    ~            & MovieLens & Gowalla & Pinterest & Yelp2018 & AMZ-Book & Dianping \\
\midrule
  ND         & {22.80({\drop \tiny -6.44\%})}   & {16.48({\drop \tiny -4.13\%})}   & {14.91({\drop \tiny -2.93\%})}  & {5.54({\drop \tiny -7.05\%})}& {3.73({\drop \tiny -1.32\%})}& {11.75({\drop \tiny -2.49\%})}\\
  ED         & {23.58({\drop \tiny -3.24\%})}   & {16.88({\drop \tiny -1.80\%})}   & {14.85({\drop \tiny -3.32\%})}  & {5.51({\drop \tiny -7.55\%})}& {3.78 ({\tiny 0.00\%})}& {11.89({\drop \tiny -1.33\%})}  \\
 GRW        &{22.39}{{\drop \tiny (-8.12\%)}}  & {16.14({\drop \tiny -6.11\%})}   & {14.77({\drop \tiny -3.84\%})}  & {5.60({\drop \tiny -6.04\%})}& {3.76({\drop \tiny -0.53\%})}& {11.84({\drop \tiny -1.74\%})} \\  

{\cykk ND+ED} & {\cykk 23.87({\drop \tiny -2.05\%})}  & {\cykk 17.08({\drop \tiny -0.64\%})}  & {\cykk 15.11({\drop \tiny -1.63\%})}  & {\cykk 5.82({\drop \tiny -2.35\%})}  & {\cykk 3.77({\drop \tiny -0.26\%})}  & {\cykk 11.98({\drop \tiny -0.58\%})}  \\
{\cykk ND+GRW} & {\cykk 22.97({\drop \tiny -5.74\%})}  & {\cykk 16.43({\drop \tiny -4.42\%})}  & {\cykk 14.98({\drop \tiny -2.47\%})}  & {\cykk 5.71({\drop \tiny -4.19\%})}  & {\cykk 3.72({\drop \tiny -1.59\%})}  & {\cykk 11.89({\drop \tiny -1.33\%})}  \\
{\cykk ED+GRW} & {\cykk 23.51({\drop \tiny -3.53\%})}  & {\cykk 16.92({\drop \tiny -1.57\%})}  & {\cykk 15.07({\drop \tiny -1.89\%})}  & {\cykk 5.79({\drop \tiny -2.85\%})}  & {\cykk 3.77({\drop \tiny -0.26\%})}  & {\cykk 11.94({\drop \tiny -0.91\%})}  \\
{\cykk ND+ED+GRW} & {\cykk 24.13({\drop \tiny -0.98\%})}  & {\cykk 17.06({\drop \tiny -0.76\%})}  & {\cykk 15.14({\drop \tiny -1.43\%})}  & {\cykk 5.87({\drop \tiny -1.51\%})}  & {\cykk 3.81({\tiny +0.79\%})}  & {\cykk 12.01({\drop \tiny -0.33\%})}  \\
\midrule[0.1pt]
  \model            & \textbf{24.37}  & \textbf{17.19}   & \textbf{15.36}   & \textbf{5.96}   & \textbf{3.78}   & \textbf{12.05}   \\
\bottomrule
\end{tabular}}
\end{table}

\textbf{Regularization Effect of Representation Uniformity.}
Previous work~\cite{wang2020understanding} identifies that contrastive learning can help to improve the \textit{uniformity} of image representations.
To study its effect in learning hash codes, we visualize the feature distributions of learned hash codes between BGCH and \model~as follows.
Firstly, we conduct the dimension reduction using T-SNE~\cite{hinton2002stochastic} to project the acquired representations into the 2-dimensional space.
The projected representations are then normalized onto the unit hypersphere, ensuring a radius of 1.
Next, we employ nonparametric Gaussian kernel density estimation (KDE)~\cite{botev2010kernel} to depict the distributions in Figure~\ref{fig:kde}.

In Figure~\ref{fig:kde}, the upper row presents BGCH with no dual feature contrastive learning, denoted by \textsl{w/o DFCL}; and the lower row corresponds to \model~with the proposed module, i.e., \textsl{w/in DFCL}.
We observe notably different feature distributions between the two scenarios.
In the \textsl{w/o DFCL} scenario, the features tend to be highly clustered in specific areas. 
Besides, the \textsl{w/in DFCL} scenario however exhibits more uniform distributions, with features residing on wider arcs across all six datasets.
The observed differences in feature distributions align with the performance improvements shown in Tables~\ref{tab:topn}-\ref{tab:full}.
This indicates that the uniformity of learned hashed codes plays a decisive role in the quality of graph hashing and ultimately impacts the prediction performance.

\textbf{Effect of Dual Feature Contrastive Learning.}
To investigate the contribution of two feature augmentations, we set two variants via disabling the contrastive learning on different stages: (1) disabling contrastive learning on intermediate continuous embeddings, denoted as \textsl{w/o $\mathcal{L}^1_{cl}$}, and (2) disabling learning on output hash codes, denoted as \textsl{w/o $\mathcal{L}^2_{cl}$}. 
As presented in Table~\ref{tab:cl_loss}, we observe that both loss terms are jointly important for achieving desired performance.

However, for the MovieLens dataset, the variant \textsl{w/o $\mathcal{L}^2_{cl}$} showed slightly better performance. 
This can be attributed to the high density of the MovieLens graph compared to other datasets.
Specifically, let the density be defined by $\frac{\left|\mathcal{V}_1\right| \times\left|\mathcal{V}_2\right|}{|\mathcal{E}|}$. 
Sparse datasets with smaller densities, such as Gowalla (0.00084), Pinterest (0.00267), Yelp2018 (0.00130), and AMZ-Book (0.00062), are significantly influenced by our feature augmentation.
On the other hand, denser datasets like MovieLens (0.04190) and Dianping (0.02210) have more graph edges for model training and are therefore less sensitive to the augmentations.
These findings highlight the impact of our dual feature augmentations, with sparser graphs benefiting more from such techniques.

\begin{figure*}[tp]

\begin{minipage}{1\textwidth}
\includegraphics[width=7.15in]{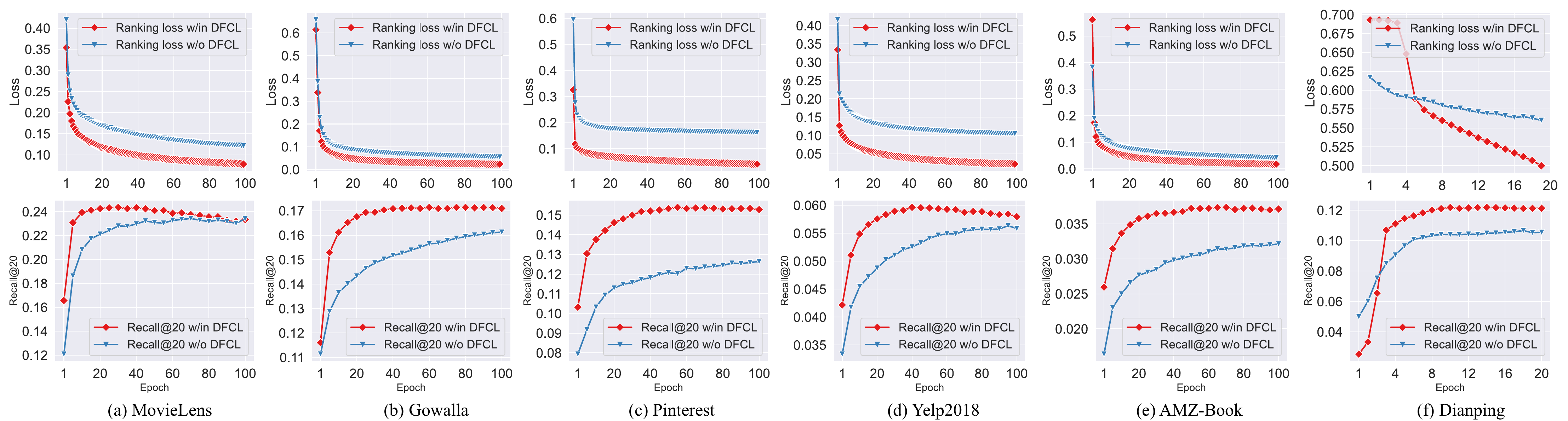}
\end{minipage} 
\caption{\newfig Illustration of model convergence in terms of (1) ranking loss (upper row) and (2) Recall metric (lower row).}
\label{fig:converge}
\end{figure*}

\begin{table}[t]
\centering
\scriptsize
\newfig
\caption{\newfig Results of different feature augmentation strategies in terms of Recall@20.}
\label{tab:cl_loss}
\setlength{\tabcolsep}{0.3mm}{
\begin{tabular}{c |c | c | c | c | c| c}
\toprule
    ~            & MovieLens & Gowalla & Pinterest & Yelp2018 & AMZ-Book & Dianping \\
\midrule
  \textsl{w/o $\mathcal{L}^1_{cl}$}        & {24.23({\drop \tiny -0.57\%})}    & {16.85({\drop \tiny -1.98\%})}   & {14.73({\drop \tiny -4.10\%})}   & {5.82({\drop \tiny -2.35\%})}   & {3.67({\drop \tiny -2.91\%})}   & {11.89({\drop \tiny -1.33\%})}   \\
  \textsl{w/o $\mathcal{L}^2_{cl}$}        & {24.56({\tiny +0.78\%})}    & {16.88({\drop \tiny -1.80\%})}   & {14.98({\drop \tiny -2.47\%})}   & {5.93({\drop \tiny -0.50\%})}   & {3.64({\drop \tiny -3.70\%})}   & {11.94({\drop \tiny -0.91\%})}   \\
\midrule[0.1pt]
  \model            & \textbf{24.37}  & \textbf{17.19}   & \textbf{15.36}   & \textbf{5.96}   & \textbf{3.78}   & \textbf{12.05}   \\
\bottomrule
\end{tabular}}
\end{table}

\textbf{Convergence Analysis.}
Lastly, we discuss the model convergence speed over all datasets and made two key observations.
We collect the metrics within the first 10\% of the whole training epochs and depict results in Figure~\ref{fig:converge}.
Firstly, our proposed dual feature contrastive learning in \model~contributes to a faster convergence to the ranking loss.
Secondly, the faster convergence facilitated by our DFCL design enables \model~to achieve its optimal performance much earlier compared to BGCH without dual feature contrastive learning. 
For instance, in MovieLens dataset, \model~already achieves its best performance with only about 20 epochs.
Similarly, for other datasets, \model~reaches its peak performance after approximately the 40$^{\text{th}}$ epoch, while BGCH without DFCL is still several epochs distant from its best performance.
In summary, these observations demonstrate that our DFCL design in \model~can effectively accelerate the model's convergence, enabling \model~to converge faster and achieve superior performance compared to its counterpart without the dual feature contrastive learning.

}

 \begin{table*}[t]
\scriptsize
\caption{Ablation study.}
\label{tab:ablation}
\setlength{\tabcolsep}{0.1mm}{
\begin{tabular}{c |c c|c c|c c|c c|c c|c c}
\toprule
 \multirow{2}*{Variant} & \multicolumn{2}{c|}{MovieLens} & \multicolumn{2}{c|}{Gowalla} & \multicolumn{2}{c|}{Pinterest} & \multicolumn{2}{c|}{Yelp2018}  &\multicolumn{2}{c|}{AMZ-Book} &\multicolumn{2}{c}{Dianping} \\
               ~  & R@20 & N@20 & R@20 & N@20 & R@20 & N@20 & R@20 & N@20 & R@20 & N@20 & R@20 & N@20\\
\midrule

 \textsl{w/o AH-TA}    &{21.23}{{\drop \notsotiny (-12.88\%)}} &{31.49}{{\drop \notsotiny (-15.37\%)}} &{15.91}{{\drop \notsotiny (-7.45\%)}} &{14.26}{{\drop \notsotiny (-16.61\%)}} &{13.77}{{\drop \notsotiny (-10.35\%)}} &{9.50}{{\drop \notsotiny (-15.25\%)}} &{5.12}{{\drop \notsotiny (-14.09\%)}} &{5.66}{{\drop \notsotiny (-8.27\%)}} &{2.98}{{\drop \notsotiny (-21.16\%)}} &{3.22}{{\drop \notsotiny (-25.98\%)}} &{\ \ 11.02}{{\drop \notsotiny (-8.55\%)}} &{7.45}{{\drop \notsotiny (-15.44\%)}}  \\

\textsl{w/o AH-RF}   &{18.38}{{\drop \notsotiny (-24.58\%)}} &{28.13}{{\drop \notsotiny (-24.40\%)}} &{13.19}{{\drop \notsotiny (-23.27\%)}} &{12.77}{{\drop \notsotiny (-25.32\%)}} &{13.05}{{\drop \notsotiny (-15.04\%)}} &{9.11}{{\drop \notsotiny (-18.73\%)}} &{4.71}{{\drop \notsotiny (-20.97\%)}} &{4.97}{{\drop \notsotiny (-19.45\%)}} &{3.24}{{\drop \notsotiny (-14.29\%)}} &{3.82}{{\drop \notsotiny (-12.18\%)}} &{\, \ 9.89}{{\drop \notsotiny (-17.93\%)}} &{7.48}{{\drop \notsotiny (-15.10\%)}}  \\

 \textsl{w/in LF}  &{22.24}{{\drop \notsotiny (-8.74\%)}} &{35.66}{{\drop \notsotiny (-4.17\%)}} &{16.24}{{\drop \notsotiny (-5.53\%)}} &{16.17}{{\drop \notsotiny (-5.44\%)}} &{13.38}{{\drop \notsotiny (-12.89\%)}} &{9.55}{{\drop \notsotiny (-14.81\%)}} &{5.23}{{\drop \notsotiny (-12.25\%)}} &{5.53}{{\drop \notsotiny (-10.37\%)}} &{3.45}{{\drop \notsotiny (-8.73\%)}} &{3.89}{{\drop \notsotiny (-10.57\%)}} &{10.56}{{\drop \notsotiny (-12.37\%)}} &{7.91}{{\drop \notsotiny (-10.22\%)}}  \\

 \textsl{w/o $\mathcal{L}_{bpr}$} &{22.51}{{\drop \notsotiny (-7.63\%)}} &{35.31}{{\drop \notsotiny (-5.11\%)}} &{15.22}{{\drop \notsotiny (-11.46\%)}} &{15.94}{{\drop \notsotiny (-6.78\%)}} &{13.46}{{\drop \notsotiny (-12.37\%)}} &{9.92}{{\drop \notsotiny (-11.51\%)}} &{5.58}{{\drop \notsotiny (-6.38\%)}} &{5.72}{{\drop \notsotiny (-7.29\%)}} &{3.50}{{\drop \notsotiny (-7.41\%)}} &{3.93}{{\drop \notsotiny (-9.66\%)}} &{10.84}{{\drop \notsotiny (-10.04\%)}} &{7.76}{{\drop \notsotiny (-11.92\%)}}  \\

\textsl{w/o $\mathcal{L}_{cl}$} &{22.43}{{\drop \notsotiny (-7.96\%)}} &{35.55}{{\drop \notsotiny (-4.46\%)}} &{14.76}{{\drop \notsotiny (-14.14\%)}} &{15.25} {{\drop \notsotiny (-10.82\%)}} &{13.79}{{\drop \notsotiny (-10.22\%)}} &{9.88}{{\drop \notsotiny (-11.86\%)}} &{5.66}{{\drop \notsotiny (-5.03\%)}} &{5.84}{{\drop \notsotiny (-5.35\%)}} &{3.20}{{\drop \notsotiny (-15.34\%)}} &{3.69}{{\drop \notsotiny (-15.17\%)}} &{11.33}{{\drop \notsotiny (-5.98\%)}} &{7.90}{{\drop \notsotiny (-10.33\%)}} \\

\midrule[0.1pt]
  \textbf{\model}   &\textbf{24.37}& \textbf{37.21}  &\textbf{17.19}& \textbf{17.10} &\textbf{15.36}& \textbf{11.21} &\textbf{5.96} & \textbf{6.17} &\textbf{3.78} & \textbf{4.35} &\textbf{12.05}& \textbf{8.81}    \\ 

\bottomrule
\end{tabular}}
\end{table*}

\subsection{\textbf{Ablation Study (RQ4)}}
\label{sec:ablation}
We evaluate the necessity of model components with Top-20 search metrics and report the results in Table~\ref{tab:ablation}.

{\textbf{Effect of Adaptive Graph Convolutional Hashing.}}
We study this model component by setting two variants, where: (1) \textsl{w/o AH-TA} only disables the \textit{topology-awareness of hashing} and uses the final encoder after all graph convolutions (similar to conventional approaches~\cite{hashgnn,hashnet}); (2) \textsl{w/o AH-RF} removes the \textit{rescaling factors}.
The results from Table~\ref{tab:ablation} results produce to the following observations:
\begin{enumerate}[leftmargin=*]
\item 
The variant \textsl{w/o AH-TA} underperforms \model.
This indicates that solely relying on the final output embeddings from the Graph Convolutional Network (GCN) framework may not adequately capture the unique latent node features necessary for effective hashing, especially considering the rich structural information present at different graph depths. 
In contrast, \model~leverages intermediate information to enrich the representations, resulting in topology-aware hashing that effectively addresses the limited expressivity of discrete hash codes.

\item In addition to topology-aware hashing, the inclusion of \textit{rescaling factors} (as introduced in Equation~(\ref{eq:rescale})) plays a crucial role in performance improvement.
The removal of these factors from \model~(variant \textsl{w/o AH-RF}) leads to significant performance decay. 
Although the computation of these factors is based on direct calculations and may not be theoretically optimal, they capture the numerical uniqueness of embeddings for subsequent hash encoding, which substantially enhances \model's prediction capability.
The \textit{determinacy} design of such factor computation is explored in detail in the following section.
\end{enumerate}


{\textbf{Design of Learnable Rescaling.}}
To learn the performance of \textit{learnable rescaling factors}, we include another variant namely \textsl{w/in LF}.
However, as shown in Table~\ref{tab:ablation}, the design of learnable rescaling factors in \textsl{w/in LF} does not achieve good performance as expected. 
One possible explanation for this outcome is that our current model does not impose strong mathematical constraints to the learnable factors ($\alpha_x$), e.g., $\alpha_x^{(l)} = \argmin({\emb{V}}^{(l)}_x$, $\alpha_x^{(l)} \emb{{Q}}^{(l)}_x)$, mainly because of its additional training complexity.
Consequently, relying solely on stochastic optimization methods, such as stochastic gradient descent (SGD), may struggle to find the optimal values for these factors.
Considering the additional search space introduced by the incorporation of learnable rescaling factors and the limitations of stochastic optimization, we argue that our deterministic rescaling method is a simple yet effective approach in practice. It strikes a balance between computational efficiency and performance, making it a preferable choice for our proposed model.

{\textbf{Effect of Multi-loss in Optimization.}}
To investigate the impact of BPR loss $\mathcal{L}_{bpr}$ and contrastive learning loss $\mathcal{L}_{cl}$, we set two variants, termed by \textsl{w/o $\mathcal{L}_{bpr}$} and \textsl{w/o $\mathcal{L}_{cl}$}, to optimize \model~separately.
These variants are applied while keeping all other model components intact. 
As shown in Table~\ref{tab:ablation}, partially using each one of $\mathcal{L}_{bpr}$ and $\mathcal{L}_{cl}$ can not yield the expected performance.
This finding validates the effectiveness of our proposed multi-loss design.
while $\mathcal{L}_{bpr}$ guides the model to assign higher prediction values to observed edges, i.e., $\emb{Y}_{x,y}=1$, than the unobserved node pair counterparts, 
$\mathcal{L}_{cl}$ helps to alleviate the data sparsity issue and promotes the uniformity of output representations.
helps address data sparsity issues and promotes the uniformity of output representations.
By jointly optimizing these two loss functions, our model \model~can learn high-quality binarized embeddings from $\mathcal{L}_{cl}$, and maintain rich relative order information regularized by $\mathcal{L}_{bpr}$ accordingly.
Hence, our multi-loss framework enables \model~to achieve superior performance in terms of both representation quality and ranking capability.

\subsection{\textbf{Resource Consumption Analysis (RQ5)}}
\label{sec:resource}
Due to the various value ranges over all six datasets, we compactly report the value ratios of \model~over the state-of-the-art hashing-based model HashGNN$_s$ in Figure~\ref{fig:tradeoff}.

\begin{figure}[t]
\begin{minipage}{0.5\textwidth}
\includegraphics[width=3.5in]{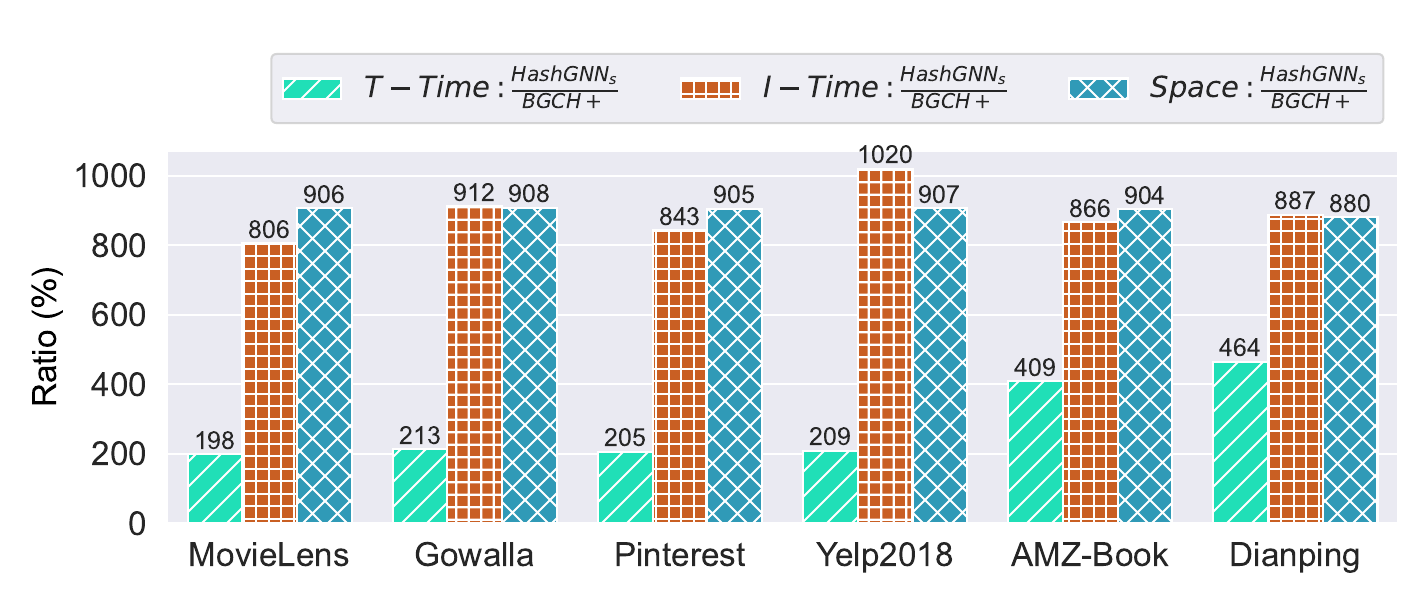}
\end{minipage} 
\caption{Resource consumption ratios.}
\label{fig:tradeoff}
\end{figure}

{\textbf{Model Training Time Cost.}}
The training time cost, represented by the metric ``\textit{T-Time}'' in Figure~\ref{fig:tradeoff}, reveals that training HashGNN$_s$ is more time-consuming compared to our proposed model \model.
This difference can be attributed to the architectural disparities between the two models.
HashGNN$_s$ utilizes the earlier Graph Convolutional Network (GCN) framework~\cite{graphsage} as the model backbone, which involves additional operations such as self-connection, feature transformation, and nonlinear activation.
On the other hand, our model, \model, follows the latest GCN framework~\cite{lightgcn} and eliminates these operations, resulting in reduced computational complexity and faster training.
Furthermore, on the two largest datasets, AMZ-Book and Dianping, the training cost ratio becomes even more pronounced, reaching approximately 4 to 4.6 times higher for HashGNN$_s$ compared to \model. This is because of the need to decrease the batch size of HashGNN$_s$ to ensure a manageable training process.

{\textbf{Online Inference Time Cost.}}
We randomly generate 1,000 queries and evaluate the computation time cost.
To ensure a fair comparison, we disable all parallel arithmetic techniques, such as MKL and BLAS, by using an open-source toolkit\footnote{\notsotiny\url{https://www.lfd.uci.edu/~gohlke/pythonlibs/}}.
Indicated by ``\textit{I-Time}'' in Figure~\ref{fig:tradeoff}, our model with Hamming distance matching generally achieves over 8$\times$ computation acceleration over HashGNN$_s$ on all datasets.
This is because, as we have explained in~\cref{sec:exp_topn}, HashGNN$_s$ randomly replaces the hash codes with their original continuous embeddings for relaxation purposes and relies on floating-point arithmetics, which sacrifices the computation acceleration provided by bitwise operations.

{\textbf{Hash Codes Memory Footprint.}}
Embedding binarization can largely reduce memory space consumption.
Compared to the state-of-the-art model HashGNN$_s$, our \model~achieves about 9$\times$ of memory space reduction for the hash codes.
As we have just explained, since HashGNN$_s$ interprets hash codes with random real-value digits, it thus requires additional cost to distinguish binary digits from full-precision ones. 
In contrast, \model~separates the storage of binarized encoding parts and corresponding rescaling factors. 
This separation allows for optimized space overhead and efficient storage of the binarized embeddings.

\subsection{\textbf{Study of Fourier Gradient Estimation (RQ6)}}
\label{sec:fs_exp}
We take our largest dataset Dianping for illustration and the analysis can be generally popularized to the other datasets.

\begin{figure}[t]
\begin{minipage}{0.5\textwidth}
\includegraphics[width=3.4in]{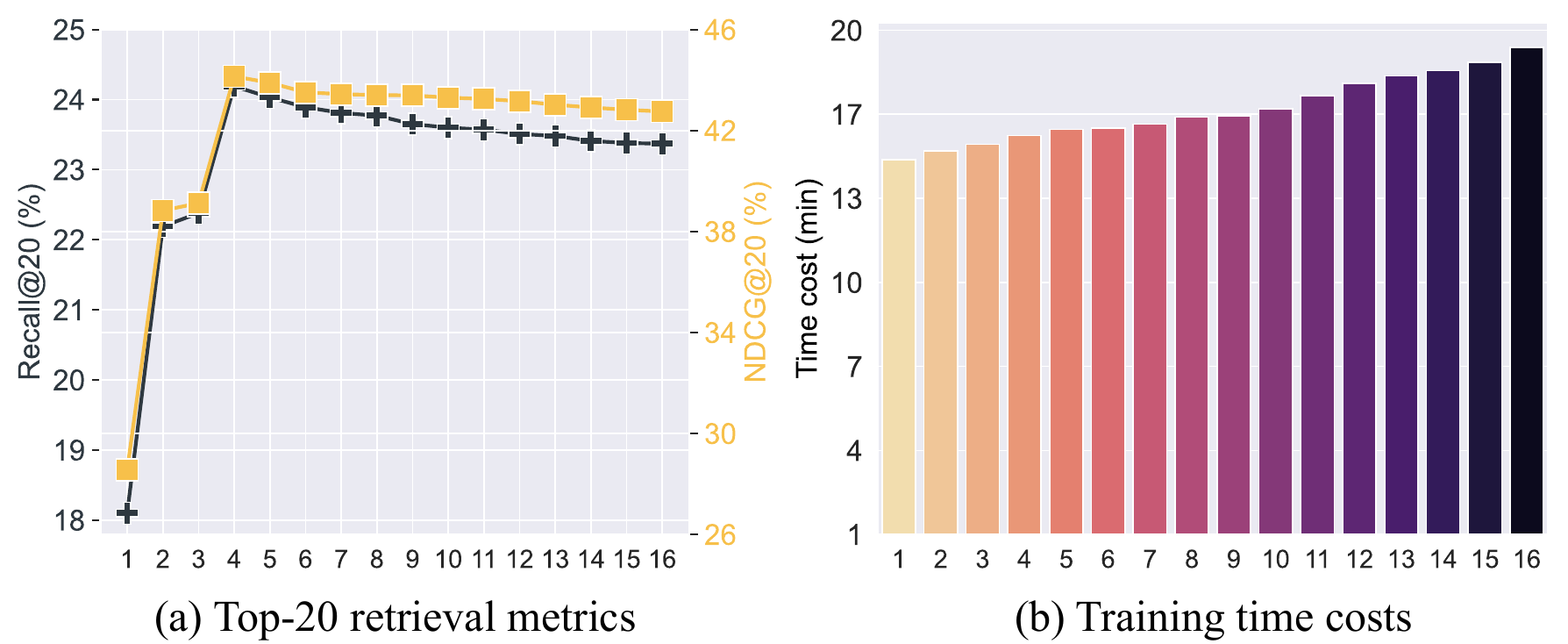}
\end{minipage} 
\caption{Fourier Series decomposition term $n$ in \model.}
\label{fig:fs_n}
\end{figure}

\begin{table}[t]
\centering
\scriptsize
\caption{Gradient estimator comparison on Recall@20.}
\label{tab:estimator}
\setlength{\tabcolsep}{0.1mm}{
\begin{tabular}{c |c | c | c | c | c| c}
\toprule
    ~            & MovieLens & Gowalla & Pinterest & Yelp2018 & AMZ-Book & Dianping \\
\midrule
  STE         & {22.14({\drop \tiny -9.15\%})}   & {15.34({\drop \tiny -10.76\%})}   & {13.57({\drop \tiny -11.65\%})}   & {5.42({\drop \tiny -9.06\%})}   & {3.41({\drop \tiny -9.79\%})}   & {11.39({\drop \tiny -5.48\%})}  \\
  Tanh         & {22.66({\drop \tiny -7.02\%})}   & {15.82({\drop \tiny -7.97\%})}   & {14.45({\drop \tiny -5.92\%})} & {5.82({\drop \tiny -2.35\%})}   & {3.43({\drop \tiny -9.26\%})}     & {11.74({\drop \tiny -2.57\%})}  \\
  SignSwish         & {23.14({\drop \tiny -5.05\%})}   & {16.70({\drop \tiny -2.85\%})}   & {14.52({\drop \tiny -5.47\%})}   & {5.67({\drop \tiny -4.87\%})}   & {3.42({\drop \tiny -9.52\%})}   & {11.57({\drop \tiny -3.98\%})}  \\
  Sigmoid         & {23.44({\drop \tiny -3.82\%})}   & {15.89({\drop \tiny -7.56\%})}   & {14.28({\drop \tiny -7.03\%})}   & {5.79({\drop \tiny -2.85\%})}   & {3.56({\drop \tiny -5.82\%})}   & {11.89({\drop \tiny -1.33\%})}  \\
  PBE         & {22.57({\drop \tiny -7.39\%})}   & {16.27({\drop \tiny -5.35\%})}   & {14.20({\drop \tiny -7.55\%})}   & {5.48({\drop \tiny -8.05\%})}   & {3.67({\drop \tiny -2.91\%})}   & {11.55({\drop \tiny -4.15\%})}  \\
\midrule[0.1pt]
  \model            & \textbf{24.37}  & \textbf{17.19}   & \textbf{15.36}   & \textbf{5.96}   & \textbf{3.78}   & \textbf{12.05}   \\
\bottomrule
\end{tabular}}
\end{table}

{\textbf{Effect of Decomposition Term $n$.}}
We vary the decomposition term $n$ from 1 to 16.
As shown in Figure~\ref{fig:fs_n}, we have two observations:
(1) The choice of the decomposition term has a significant impact on the retrieval quality.
Theoretically, larger values of $n$ can provide more accurate gradient estimations. 
However, in practice, excessively large $n$ may lead to overfitting. Therefore, it is advisable to choose a moderate value, such as $n$ = 4 in Figure~\ref{fig:fs_n}(a), to maximize model performance.
(2) As we vary $n$ from 1 to 16, the training time per iteration of \model~gradually increased. 
This observation aligns with our complexity analysis in~\cref{sec:complexity}, where we identified that the training cost is primarily associated with other modules like graph convolutional hashing, rather than the gradient estimation process.

{\textbf{Comparison with Other Gradient Estimators.}}
We include several recent gradient estimators, i.e., \textit{Tanh-like}~\cite{qin2020forward,gong2019differentiable}, \textit{SignSwish}~\cite{darabi2018bnn}, \textit{Sigmoid}~\cite{sigmoid}, and \textit{projected-based estimator}~\cite{RBCN} (denoted as PBE).
(1) The results summarized in Table~\ref{tab:estimator} clearly demonstrate the superiority of our method over $\sign(\cdot)$ function approximation in gradient estimation.
As we have briefly explained, most existing estimators, which employ the \textit{visually similar} function approximation to $\sign(\cdot)$, generally provide better gradient estimation than Straight-Through Estimator (STE).
(2) However, for bipartite graphs with high sparsity, e.g., Gowalla (0.00084) and AMZ-Book (0.00062), graph-based models may struggle to collect sufficient structural information for effective training of hash codes.
With limited training samples, these \textit{theoretically irrelevant} estimators may fail to rectify optimization deviations effectively, leading to noticeable performance gaps compared to our Fourier Series decomposition estimator.


\section{{Related Work}}
\label{sec:work}
{\textbf{Graph Convolution Network (GCN)}.}
Early research primarily studies the graph convolutions in the \textit{spectral domain}, such as Laplacian eigen-decomposition~\cite{bruna2013spectral} and Chebyshev polynomials~\cite{defferrard2016convolutional}.
One major issue is that they usually suffer from high computationally expensive. 
To tackle this problem, \textit{spatial-based} GCN models are proposed to re-define the graph convolution operations by aggregating the embeddings of neighbors to refine and update the target node's embedding.
Due to its scalability to large graphs, spatial-based GCN models are successfully applied in various applications~\cite{lightgcn,graphsage,chen2023topological,chen2024deep,wu2023survey,chen2024,zhang2024influential}. 
Despite their effectiveness in embedding latent features for graph nodes, these models usually suffer from inference inefficiency due to the high computational cost associated with calculating similarities between continuous embeddings~\cite{hashgnn}.
To address this issue, \textit{learning to hash} has emerged as a viable solution.

\

{\textbf{Learning to Hash}.}
Learning to hash models have shown great promise in achieving computational acceleration and storage reduction.
By employing similarity-preserving hashing techniques, they can efficiently map high-dimensional dense vectors to a low-dimensional Hamming space, facilitating downstream tasks.
{\cykk Locality Sensitive Hashing (LSH)~\cite{lsh,charikar2002similarity} uses a series of hash projections to collect similar data points to the same or nearby ``buckets'' with high probability.}
More recent research has focused on integrating deep neural network architectures to improve the mode performance\cite{wang2017survey}. 
This has led to a series of follow-up studies for various asks such as fast retrieval of images~\cite{qin2020forward,lin2017towards,hashnet}, documents~\cite{li2014two,chen2022effective,hihpq}, categorical information~\cite{kang2021learning}, and e-commerce products~\cite{zhang2017discrete,chen2022learning}.

To leverage hashing techniques with GCNs, the recent work HashGNN~\cite{hashgnn} investigates \text{learning to hash} for online matching and recommendation.
Specifically, HashGNN combines the GraphSage~\cite{graphsage} as the embedding encoder and applies learning to hash methods to obtain binary encodings.
Its hash encoding process only proceeds at the end of multi-layer graph convolutions, i.e., using the aggregated output of GraphSage for embedding binarization. 
However, this fails to capture intermediate semantics from nodes' different layers of receptive fields~\cite{kipf2016semi}.
HashGNN utilizes the Straight-Through Estimator (STE)~\cite{bengio2013estimating} to assume all gradients of $\sign(\cdot)$ as 1 during backpropagation.
But the integral of 1 is a certain linear function other than the $\sign(\cdot)$, whereas this may lead to inconsistent optimization directions in the model training.
To tackle these issues, our model \model~is proposed.

\

{\textbf{Graph Contrastive Learning.}} 
Graph contrastive learning has recently emerged as a prominent research direction, drawing inspiration from the success of contrastive learning in visual representation tasks~\cite{he2020momentum,chen2020simple}.
In the context of graph data, contrastive learning typically requires the explicit application of data augmentation techniques.
Traditionally, graph data augmentation methods have focused on manipulating the graph structures themselves, employing strategies such as node dropout, edge dropout, and graph random walk. However, these approaches may introduce biases that can affect the quality and integrity of the learned representations~\cite{zhang2022costa}.
To mitigate these challenges, an alternative approach is to perform data augmentation in the feature space rather than directly modifying the input space~\cite{feng2021survey}.
This can be achieved through techniques such as feature perturbation~\cite{zhang2022costa,yu2022graph}. 
By augmenting the feature representations, graph contrastive learning aims to maximize the agreement between two augmented views of the same graph in the latent space, while simultaneously ensuring the differentiation of representations for different nodes.
Notably, recent studies have demonstrated the superior performance of graph contrastive learning in various graph-related tasks~\cite{zhu2020deep,yu2022graph,wu2021self,zhang2023mitigating}, which motivates us to further study the problem of graph contrastive hashing.

{\cyk
\section{{Conclusion}}
\label{sec:con}
In this paper, we revisit the learning to hash for efficient Hamming space search over graph structure data and propose \model~for performance improvement.
Compared to its predecessor, \model~is further equipped with a novel dual feature contrastive learning paradigm, which operates on the latent features instead of the input graphs.
Such design well enhances the robustness of learned hash codes against variations and thereby promotes the extraction of graph semantics in hash encoding.
The empirical analyses over six real-world datasets demonstrate that the proposed method consistently outperforms existing hashing-based models while providing an alternative to full-precision models in scenarios with limited resources.

}

\section{Acknowledgments}
The research presented in this paper was partially supported by the Research Grants Council of the Hong Kong Special Administrative Region, China (CUHK 14222922, RGC GRF 2151185).
Yixiang Fang was supported by NSFC under Grant 62102341, Guangdong Talent Program under Grant 2021QN02X826, and Shenzhen Science and Technology Program under Grants JCYJ20220530143602006, ZDSYS 20211021111415025. 
Chenhao Ma was supported by NSFC under Grant 62302421, Basic and Applied Basic Research Fund in Guangdong Province under Grant 2023A1515011280, and the Guangdong Provincial Key Laboratory of Big Data Computing, The Chinese University of Hong Kong, Shenzhen.

\balance
\bibliographystyle{IEEEtran}
\bibliography{ref}




\end{document}